\documentclass[11pt]{article}
\usepackage{amsthm}

\usepackage[margin=1in]{geometry}

\usepackage{graphicx}
\usepackage{amsfonts,amsmath,amssymb}

\usepackage{algorithm}
\usepackage[noend]{algpseudocode}

\newtheorem{theorem}{Theorem}[section]
\newtheorem{observation}[theorem]{Observation}
\newtheorem{proposition}[theorem]{Proposition}
\newtheorem{lemma}[theorem]{Lemma}
\newtheorem{corollary}[theorem]{Corollary}
\newtheorem{definition}[theorem]{Definition}
\newtheorem*{thma}{Theorem}

\newcommand{\oa}{\overline{\mathcal A}}

\newcommand{\bE}{\ensuremath{\mathbf{E}}}

\DeclareMathOperator{\var}{var}
\DeclareMathOperator{\poly}{poly}
\def\myparagraph#1{\vspace{2pt}\noindent{ #1~~}}

\newtheorem*{rep@theorem}{\rep@title}
\newcommand{\newreptheorem}[2]{%
\newenvironment{rep#1}[1]{%
 \def\rep@title{#2 \ref{##1}}%
 \begin{rep@theorem}}%
 {\end{rep@theorem}}}
\makeatother

\newreptheorem{theorem}{Theorem}

\begin{document}

\title{Oblivious resampling oracles and parallel algorithms for the Lopsided Lov\'{a}sz Local Lemma}
\author{David G. Harris\thanks{Department of Computer Science, University of Maryland, 
College Park, MD 20742. 
Email: \texttt{davidgharris29@gmail.com}}}

\date{}

\maketitle

\begin{abstract}
The Lov\'{a}sz Local Lemma (LLL) shows that, for a collection of ``bad'' events $\mathcal B$ in a probability space which are not too likely and not too interdependent, there is a positive probability that no events in $\mathcal B$ occur.  Moser \& Tardos (2010) gave sequential and parallel algorithms which transformed most applications of the variable-assignment LLL into efficient algorithms. A framework of Harvey \& Vondr\'{a}k (2015) based on ``resampling oracles'' extended this to sequential algorithms for other probability spaces satisfying a generalization of the LLL known as the  Lopsided Lov\'{a}sz Local Lemma (LLLL).

We describe a new structural property which holds for most known resampling oracles, which we call ``obliviousness.''  Essentially, it means that the interaction between two bad-events $B, B'$ depends only on the randomness used to resample $B$, and not the precise state within $B$ itself.

This property has two major consequences. First, combined with a framework of Kolmogorov (2016), it leads to a unified parallel LLLL algorithm, which is faster than previous, problem-specific algorithms of Harris (2016) for the variable-assignment LLLL and of Harris \& Srinivasan (2014) for permutations. This gives the first RNC algorithms for rainbow perfect matchings and rainbow hamiltonian cycles of $K_n$. 

Second, this property allows us to build LLLL probability spaces from simpler ``atomic'' events. This gives the first resampling oracle for rainbow perfect matchings on the complete $s$-uniform hypergraph $K_n^{(s)}$, and the first commutative resampling oracle for hamiltonian cycles of $K_n$.
\end{abstract}


\maketitle

This is an extended version of a paper which appeared in the ACM-SIAM Symposium on Discrete Algorithms (SODA) 2019.

\section{The Lov\'{a}sz Local Lemma and its algorithms}
The Lov\'{a}sz Local Lemma (LLL) is a fundamental probabilistic tool which shows that for a probability space $\Omega$ with a finite set $\mathcal B$ of $m$ ``bad'' events, then as long as the bad-events are not too interdependent (in a certain technical sense) and are not too likely,  there is a positive probability no events in $\mathcal B$ occur.  The simplest form of the LLL, known as the \emph{symmetric LLL}, can be stated as follows: if every bad-event $B$ has $\Pr_{\Omega}(B) \leq p$ and is dependent with at most $d$ others, where $e p d < 1$, then there is a positive probability that none of the bad-events occur.

Most combinatorial applications of the LLL use a relatively simple probability space, which we call the \emph{variable-assignment LLL.} This setting has $n$ independent variables $X_1, \dots, X_n$, and each bad-event $B$ is a boolean function of a subset of these variables denoted $\var(B)$. Bad-events $B, B'$ are dependent (written $B \sim B'$) iff $\var(B) \cap \var(B') \neq \emptyset$. Moser \& Tardos \cite{moser-tardos} introduced a remarkably simple algorithm for this setting, which we refer to as the \emph{MT algorithm}:
\begin{algorithm}[H]
\centering
\begin{algorithmic}[1]
\State Draw each variable independently from the distribution $\Omega$
\While{there is a true bad-event on $X$}
\State  Choose a true bad-event $B$ arbitrarily
\State Resample $\var(B)$ according to the distribution $\Omega$
\EndWhile
\end{algorithmic}
\caption{The MT algorithm}
\end{algorithm}

Moser \& Tardos \cite{moser-tardos} showed that this algorithm terminates quickly whenever the symmetric LLL criterion (or a more general asymmetric LLL criterion) is satisfied. Later work \cite{pegden, kolipaka, harris-llll} showed that it terminates under more general criteria. See Appendix~\ref{app-kolm} for background on the LLL and MT algorithm.

Note that the MT algorithm requires a subroutine to find a bad-event $B$ which is true on the current configuration $X$ (if any). We refer to this as a \emph{Bad-Event Checker (BEC)}. The simplest implemention of this is to loop over all bad-events and test them one by one, which would have a run-time on the order of $m$. The run-time of the MT algorithm can often be polynomial in $n$ and independent of $m$ if a more-efficient BEC is used \cite{hss, harris-srin2}.

\subsection{The Lopsided Lov\'{a}sz Local Lemma}
In \cite{erdos-spencer}, Erd\H{o}s \& Spencer noted that positive correlation among bad-events (again, in a certain technical sense) is as good as independence for the LLL. This generalization has been referred to as the \emph{Lopsided Lov\'{a}sz Local Lemma (LLLL)}. We say $B, B'$ are \emph{lopsidependent} and write $B \sim B'$ if $B, B'$ are neither independent nor positively correlated in this sense. (Formal definitions are provided later in Section~\ref{framework-sec}.)

Although the variable-assignment LLL covers the vast majority of applications in combinatorics, the LLLL is also used occasionally. For example, the original application of the LLLL used a probability space on permutations to construct Latin transversals for certain types of arrays  \cite{erdos-spencer}. Other applications include hamiltonian cycles on $K_n$   \cite{albert}, perfect matchings of $K_{n}$ \cite{lu-szekeley2}, perfect matchings of the complete $s$-uniform hypergraph $K_n^{(s)}$ \cite{mohr}, and spanning trees of $K_n$ \cite {mohr}.

The variable-assignment setting provides one of the simplest forms of the LLLL. Here, as before, there are independent variables $X_1, \dots, X_n$.  Instead of allowing arbitrary boolean functions of the variables, each bad-event should be a monomial function, i.e. of the form
$$
B \equiv X_{i_1} = j_1 \wedge \dots \wedge X_{i_k} = j_k
$$

For the LLL, we would have $B \sim B'$ if the bad-events $B$ and $B'$ share some common variable, i.e. $i_t = i'_{t'}$. For the LLLL, the (lopsi)dependency relation is more restricted: we have $B \sim B'$ if $B$ and $B'$ \emph{disagree} on some common variable, i.e. $i_t = i'_{t'}$ and $j_t \neq j'_{t'}$.

Moser \& Tardos showed that their algorithm applies to the variable-assignment LLLL setting.  In \cite{harris-srin-perm}, Harris \& Srinivasan developed an algorithm similar to the MT algorithm for the probability space of random permutations, which includes the Latin transversal application of \cite{erdos-spencer}. Extending these problem-specific algorithms, Harvey \& Vondr\'{a}k \cite{harvey} developed a general framework based on a ``resampling oracle'' $\mathfrak R$ for the probability space. We will define this formally in Section~\ref{framework-sec}, but, intuitively this is a randomized algorithm which, given some state $u$ with some bad-event $B$ true on $u$, attempts to ``rerandomize'' the configuration in a ``local'' way to fix $B$. This is similar to the way that the MT algorithm resamples the variables involved in $B$. Given this resampling oracle, the following Algorithm~\ref{gen-alg}  can be used to find a configuration avoiding the bad-events:
\begin{algorithm}[H]
\centering
\caption{A general resampling algorithm}
\label{gen-alg}
\begin{algorithmic}[1]
\State Draw the state $u$ from the distribution $\Omega$
\While{some bad-event $B$ is true on $u$}
\State Select, according to some specified rule, some $B$ true on $u$
\State Update $u \leftarrow \mathfrak R_B(u)$
\EndWhile
\end{algorithmic}
\end{algorithm}

These results have led to constructive counterparts to combinatorial results involving spanning trees and matchings of $K_n$ (both discussed in \cite{harvey}) and hamiltonian cycles of $K_n$ (subsequently developed in \cite{harvey2}). A further line of research has extended Algorithm~\ref{gen-alg}, and variants, to other spaces which do not directly  correspond to the LLLL \cite{achlioptas, achlioptas2, sinclair, harris-srin-partial}.

We note that the choice of which bad-event to select in line (3) of Algorithm~\ref{gen-alg} is much more constrained than for the MT algorithm. Only a limited number of possibilities work in general, such as selecting $B$ with smallest index, whereas the MT algorithm allows nearly complete freedom. In \cite{kolmogorov}, Kolmogorov showed that a number of resampling oracles (including variable-assignment, permutations, and perfect matchings of $K_n$) satisfy an additional property known as \emph{commutativity}. In such cases, Algorithm~\ref{gen-alg} also allows an arbitrary choice of which bad-event to select.  Kolmogorov \cite{kolmogorov} and Iliopoulos \cite{ilio} further showed that this property has powerful algorithmic consequences, including parallel algorithms,  efficient BEC's, and bounds on the output distribution at the termination of  Algorithm~\ref{gen-alg}.

\subsection{Parallel algorithms}
Moser \& Tardos also presented a simple parallel version of their resampling algorithm. This parallel algorithm requires a slightly stronger criterion, which we refer to as \emph{$\epsilon$-slack}; for instance, the symmetric LLL requires $e p (1+\epsilon) d \leq 1$; if this satisfied, then it terminates after $O(\frac{\log m}{\epsilon})$ rounds with high probability.\footnote{We say that an event occurs with high probability (abbreviated whp), if it has probability at least $1 - n^{-\Omega(1)}$.} On a EREW PRAM, it has overall runtime $O(\frac{\log^3 m}{\epsilon})$.  We summarize the algorithm as follows:
\begin{algorithm}[H]
\centering
\begin{algorithmic}[1]
\State Draw each variable independently from the distribution $\Omega$
\While{there is a true bad-event on $X$}
\State  Select a maximal independent set (MIS) $I$ of true bad-events
\State Resample, in parallel, $\bigcup_{B \in I} \var(B)$
\EndWhile
\end{algorithmic}
\caption{The parallel MT algorithm}
\end{algorithm}

Haeupler \& Harris \cite{hh} showed that the parallel MT algorithm could be implemented in time $O(\frac{\log^3 n}{\epsilon})$ (avoiding dependence on $m$) and gave an alternative parallel algorithm in time $O(\frac{\log^2 m}{\epsilon})$. The parallel MT algorithm can also usually be implemented even for more general LLL criteria, including the asymmetric LLL and Shearer's LLL criterion  \cite{kolipaka}.

(In some computational models, multiple processors can write to a memory cell simultaneously and the runtimes can often reduced by logarithmic factors.  For simplicity,  we will be conservative and use only the EREW PRAM model throughout this paper. We say that an algorithm is in $RNC^k$ if it runs in $\tilde O(\log^k n)$ time and $\poly(n)$ processors whp on an EREW PRAM.)

The parallel MT algorithm leads in a straightforward way to distributed graph algorithms in $O( \frac{\log^2 m}{\epsilon} )$ communication rounds. There has been extensive research into obtaining faster distributed and parallel LLL algorithms; some of these algorithms require significantly stronger (but still local) conditions on the dependency $d$ and probability $p$ of the bad-events \cite{pettie, mohsen2, hgk}. Brandt et al. \cite{brandt} showed that generic distributed LLL algorithms require $\Omega(\log \log n)$ rounds.

Frustratingly, although the sequential MT algorithm works for the variable-assignment LLLL just as it does for the variable-assignment LLL, this is not true of the parallel MT algorithm. There have been only a handful of parallel algorithms for the LLLL, such as the variable-assignment LLLL algorithm of Harris \cite{harris-llll} and the permutation LLL algorithm of Harris \& Srinivasan \cite{harris-srin-perm}.

In \cite{kolmogorov} Kolmogorov proposed a general framework for constructing parallel LLLL algorithms via resampling oracles, which can be summarized as follows:
\begin{algorithm}[H]
\centering
\begin{algorithmic}[1]
\State Draw the state $u$ from the distribution $\Omega$
\While{there is a true bad-event on $u$}
\State{Set $V$ to be the set of currently-true bad-events}
\While{$V \neq \emptyset$}
\State Select, arbitrarily, some bad-event $B \in V$
\State Update $u \leftarrow \mathfrak R_B (u)$
\State Remove from $V$ all bad-events $B'$ such that either (i) $B'$ is no longer true; or (ii) $B' \sim B$
\EndWhile
\EndWhile
\end{algorithmic}
\caption{Kolmogorov's framework for parallel resampling algorithms}
\label{kolm-alg}
\end{algorithm}

Each iteration of the loop of lines (3) --- (7) is called a \emph{round}. Kolmogorov showed that, when the resampling oracle $\mathfrak R$ is commutative, then Algorithm~\ref{kolm-alg} terminates whp after $O(\log n)$ rounds. We emphasize this is a \emph{sequential} algorithm, which is in fact a version of Algorithm~\ref{gen-alg}.

 If a single round can be simulated in polylogarithmic time, then this yields an RNC algorithm.  In almost every setting where a parallel LLLL algorithm is known (including all the ones in this paper), the resampling oracle is commutative and the parallel algorithm is an implementation of Kolmogorov's framework.

This makes partial progress to a general parallel LLLL algorithm; however, there remain two significant hurdles. The most straightforward of these is a parallel implementation of $\mathfrak R$. This is trivial for the variable-assignment LLL: if bad-events $B, B'$ are both selected for resampling, then $\var(B)$ and $\var(B')$ must be disjoint and the resamplings can be executed simultaneously. For other probability spaces, it is not clear how to resample without ``locking'' the state.

The second and much more fundamental hurdle is that the LLLL resampling process is inherently sequential in a way that the LLL is not. For the LLLL (but not the LLL) it is possible that two bad-events $B, B'$ are currently true, and $B \not\sim B'$, and resampling $B$ makes $B'$ false. We say in this case that \emph{$B$ fixes $B'$}. Because of this possibility, $B$ and $B'$ cannot be resampled simultaneously; one must select (arbitrarily) one of the two bad-events to resample first, and then only resample the second one if it still remains true. One critical challenge for LLLL algorithms is to simulate in parallel the process of resampling the bad-events \emph{in sequence.}

The parallel LLLL algorithms of Harris \cite{harris-llll} and Harris \& Srinivasan \cite{harris-srin-perm} overcome these hurdles to a limited extent. However they still suffer from a number of shortcomings. Although they run in polylogarithmic time, the exponent is quite high (and is not computed explicitly). They also require additional structure, such as having bad-events which involve a polylogarithmic number of variables. Finally, and perhaps most seriously, these algorithms are highly tailored to a single probability space. They are reminiscent of the situation for LLL algorithms before the framework of Harvey \& Vondr\'{a}k \cite{harvey}: specialized algorithms with ad-hoc  analysis.

\subsection{Our contribution and overview}
We identify a new property of resampling oracles that we refer to as \emph{obliviousness}. To summarize, suppose  we have two bad-events $B, B'$ with $B \not \sim B'$, and a state $u \in B \cap B'$. The obliviousness property states that \emph{whether $B$ fixes $B'$ depends solely on the randomness used to resample $B$, and not on the state $u$ itself.}  This framework is developed in Section~\ref{framework-sec}. We find it remarkable that so many LLLL probability spaces, even the non-commutative ones, have oblivious resampling oracles: this includes variable-assignment, permutations, perfect matchings of $K_n$, perfect matchings of the hypergraph $K_n^{(s)}$,  hamiltonian cycles of $K_n$, and spanning trees of $K_n$.

\textbf{A unified parallel algorithm.} Obliviousness allows us to sidestep the second major hurdle to a parallel LLLL algorithm. It reduces the possibility of $B$ fixing $B'$ to a \emph{pairwise} phenomenon: we only need to know the resampling action chosen for $B$, \emph{not the present state (which may be changing during other resampling actions).} The space of sequential resamplings can thus be represented in a simple graph structure, allowing us to efficiently find a valid sequence.

To implement this sequence in parallel, we encode $\mathfrak R$  as a monoid action. Specifically, $\mathfrak R_B$ can be interpreted as a randomly-chosen monoid element $r_B$ acting on the current state $u$. In this way, resampling multiple bad-events $B_1, \dots, B_s$ can be interpreted algebraically as the product $r_{B_s} \dots r_{B_1} u$. This is easily parallelized by the associativity of monoidal multiplication.

We summarize our generic parallel LLLL algorithm as follows:
\begin{theorem}[Informal]
Suppose that $e p d (1+\epsilon) \leq 1$ holds for any LLLL probability space with an appropriate parallelizable resampling oracle. Then there is a parallel algorithm in time $O( \frac{\log^4 n}{\epsilon} )$ to find a state avoiding $\mathcal B$.
\end{theorem}

We summarize some notable applications of this algorithm.
\begin{enumerate}
\item  Suppose we have a $k$-SAT instance on $n$ variables and $m$ clauses, in which each variable appears in at most $L \leq \frac{2^{k+1} (1 - 1/k)^k}{(k-1)(1 + \epsilon)} - \frac{2}{k}$ clauses. There is an $RNC^4$ algorithm to find a satisfying assignment.
\item   For an integer $c \geq 2$, suppose that $H$ is a $k$-uniform hypergraph $H$ where each vertex appears in at most $L = \frac{c^k (1-1/k)^{k-1}}{k (c-1)(1+\epsilon)}$ edges. There is an randomized algorithm in $O( \frac{\log^3 n}{\epsilon})$ rounds  for the LOCAL distributed computing model to find a proper vertex $c$-coloring of $H$.
\item  Suppose that $A$ is an $n \times n$ matrix whose entries are labeled by colors and each color appears in at most $\Delta$ entries. For $\Delta \leq 0.105 n$, there is an $RNC^4$ algorithm to find a Latin transversal of $A$. For $\Delta \leq n \Bigl( \frac{(s-1)!}{2 e (1+\epsilon) s} \Bigr)^{1/(s-1)}$ there is an $RNC^4$ algorithm to find a transversal of $A$ where color appears at most $s$ times.
\item  Suppose that we have  an edge-coloring of $K_n$ where each color appears  on at most $\Delta$ edges. If $\Delta \leq 0.105 n$ and $n$ is even, there is a $RNC^4$ algorithm to find a rainbow perfect matching.  If $\Delta \leq 0.026 n$, there is an $RNC^4$ algorithm to find a rainbow hamiltonian cycle.
\end{enumerate}

Versions of the first two results with slightly worse parameters can be derived from the variable-assignment LLL and parallel MT algorithm.  Previous slower RNC algorithms are known for the third result. We are not aware of any RNC algorithms comparable with the fourth result; this answers open problems posed by Kolmogorov \cite{kolmogorov} and Harvey \& Liaw \cite{harvey2}.

\textbf{A new resampling framework.} Beyond its direct algorithmic impact, obliviousness can simplify a number of resampling oracle constructions. Most LLLL probability spaces come from a set of relatively simple ``atomic events.'' For example, in the space of uniform permutations, these are events of the form $\pi(x) = y$. A bad-event $B$ is then taken to be a conjunction of atomic events.

It is intuitively clear that the resampling oracle for the atomic events in some sense ``generates'' the resampling oracle for $\mathcal B$. A formal description of this has been elusive. To illustrate the difficulty, consider a bad-event $B = A_1 \cap A_2$ and a configuration $u \in B$, where $A_1, A_2$ are atomic events. We would like to resample $B$ by resampling $A_1$ and then resampling $A_2$. In order to obtain the correct probability distribution, we must condition on $A_2$ remaining true after resampling $A_1$. For a general resampling oracle, this conditioning step might distort the probability distribution of $u$ in an unmanageable way. But for an oblivious resampling oracle, we are guaranteed that \emph{conditioning on $A_2$ remaining true retains an independent, uniform distribution for $u$ itself.}

We derive a simple list of axioms required for an oblivious resampling oracle \emph{for the atomic events only}; these automatically lead to a resampling oracle for $\mathcal B$.  Beyond the fact that this gives new algorithmic results, this greatly simplifies many proofs and constructions for existing resampling oracles.  We highlight a few results:
\begin{enumerate}
\item We get a commutative resampling oracle, and parallel algorithms, for the space of hamiltonian cycles of $K_n$.
\item We get a resampling oracle for the space of perfect matchings of the complete hypergraph $K_n^{(s)}$. This leads to efficient  (sequential)  algorithms corresponding to non-constructive results on rainbow hypergraph matchings shown by Lu, Mohr, \& Sz\'{e}k\'{e}ly \cite{mohr}.
\end{enumerate}

\subsection{Outline}
In Section~\ref{framework-sec}, we formally define  the LLLL in terms of resampling oracles. We provide a new framework which is more algebraic compared to the probabilistic formulation originally developed in \cite{harvey}. We define the properties needed for resampling oracles, including commutativity and the new property of obliviousness. We also discuss the method for generating LLLL-compatible probability spaces from atomic events.

In Section~\ref{lfmis-sec}, we describe a new graph algorithm needed for our parallel LLLL algorithm. This computes a structure which is similar to a lexicographically-first MIS (LFMIS), but generalized to directed graphs. This plays a similar role to the MIS in the parallel MT algorithm, but respects the sequential ordering of the bad-events. We show that, for a random vertex order,  this LFMIS can be computed efficiently in $O(\log^2 n)$ rounds by a simple greedy parallel algorithm adapted from Blelloch, Fineman \& Shun \cite{greedy2} for undirected graphs. This is a pure graph theory problem which does not directly involve the LLLL, and may be of independent interest.

In Section~\ref{alg-sec}, we describe our generic LLLL algorithm in terms of a resampling oracle from the framework of Section~\ref{framework-sec}.

In Section~\ref{var-llll-sec}, we analyze the variable-assignment LLLL. We show how the simple resampling oracle (which is just to resample variables from the original distribution) fits into the formal framework of Section~\ref{framework-sec}. We provide a few example applications, to $k$-SAT and hypergraph coloring.

In Section~\ref{other-llll-sec}, we describe a few other more ``exotic'' LLLL spaces, including random permutations, hamiltonian cycles, and perfect matchings. We discuss a few applications, including to strong coloring and a number of Latin transversal problems.

\subsection{Notation}
Throughout, we let $[n]$ denote the set $\{1, \dots, n \}$. For a probability space $\Omega$ over a ground set $U$, we say that $u \approx \Omega$ if $u$ is a random variable drawn according to distribution $\Omega$. We define $\Omega[u]$ to be the probability mass of $u$, and we define $\text{Support}(\Omega)$ to be the set of values $u \in U$ with $\Omega[u]> 0$. 

For any $V \subseteq U$ we define $\Omega[V] = \Pr_{u \approx \Omega}(u \in V) = \sum_{v \in V} \Omega[v]$.  We also define $\Omega | V$ to be the conditional distribution on $V$, i.e. $(\Omega|V)[v] = \Omega[v] / \Omega[V]$ for $v \in V$.

For two random variables $X, Y$, we say $X \approx Y$ if $X, Y$ follow the same distribution. For any set $X$, we define $\text{Unif}(X)$ to be the uniform distribution on $X$.

For $s \geq 2$, we let $K_n^{(s)}$ denote the complete $s$-uniform hypergraph on vertex set $[n]$. For $s = 2$ (the complete graph), we also write $K_n = K_n^{(2)}$. We say that $M$ is a perfect matching of $K_n^{(s)}$ if it is a partition of $[n]$ into exactly $n/s$ classes of size $s$. Whenever we refer to the set of perfect matchings of $K_n^{(s)}$, we will assume implicitly that $s$ divides $n$.

We define $S_n$ to be the symmetric group on $n$ letters, viewed concretely as the set of bijections on ground set $[n]$. We write $(a \ b)$ for the transposition swapping $a$ and $b$. We also write $\sigma_1 \sigma_2$ for the functional composition $\sigma_1 \circ \sigma_2$, that is, the function sending $x$ to $\sigma_1(\sigma_2(x))$. 

For subsets $A, B$ of an algebraic structure $G$, we let $A B$ denote the product set $A B = \{ a  b \mid a \in A, b \in B \}$. Similarly, for $b \in G, A \subseteq G$ we write $b A = \{ b a \mid a \in A\}$ and $A b = \{ a b  \mid a \in A \}$. 

 For a directed graph $G = (V,E)$ and a vertex $v \in V$, we define the out-neighborhood $N^{\text{out}}(v) = \{ w \mid (v,w) \in E \}$ and the out-degree of $v$ is the cardinality of this set. Similarly we define the in-neighborhood $N^{\text{in}}(v) = \{ w \mid (w, v) \in E \}$, and the in-degree of $v$ is the cardinality of this set.
 
\section{The LLLL and resampling oracles}
\label{framework-sec}

In this section, we will formally define the LLLL and how to construct a resampling oracle for it, in the sense of Harvey \& Vondr\'{a}k \cite{harvey}. We note that Erd\H{o}s \& Spencer \cite{erdos-spencer} describes an alternate, probabilistic interpretation of the LLLL, which is slightly more general. Since this is technical to describe and we will never use this interpretation, we will not discuss this here.

Constructions based on the LLLL typically have two phases. First, we choose a large collection of highly-structured ``generic'' bad-events in a probability space, equipped with an appropriate lopsidependency relation and a resampling oracle. For example, in the variable-assignment LLLL setting, the underlying probability space is a cartesian product space with $n$ independent variables and the generic bad-events are the monomial functions of the form $X_{i_1} = j_1 \wedge \dots \wedge X_{i_k} = j_k$ for arbitrary values $k, (i_1, j_1), \dots, (i_k, j_k)$. For the permutation setting, the underlying probability space is the uniform distribution on $S_n$ and the generic bad-events have the form $\pi(x_1) = y_1 \wedge \dots \wedge \pi(x_k) = y_k$ for arbitrary values $k, (x_1, y_1), \dots, (x_k, y_k)$. 

It is impossible to avoid \emph{all} the generic bad-events. The second phase of the LLLL is to select some problem-specific, more-or-less ``random'', subset of the generic bad-events. For example, if we wish to satisfy a given $k$-SAT formula, then for each clause $X_{i_1} = j_1 \vee \dots \vee X_{i_k} = j_k$, we would have in $\mathcal B$ the bad-event $X_{i_1} = 1-j_1 \wedge \dots \wedge X_{i_k} = 1 - j_k$, which is one of the generic bad-events.

In order to show that the LLLL applies, and that Algorithm~\ref{gen-alg} converges to an assignment avoiding $\mathcal B$, we must show two things: first,  that the resampling oracle works properly on the generic set of bad-events containing $\mathcal B$. Second, that the specific chosen subset $\mathcal B$ has its probabilities and dependencies sufficiently small; for example,  each bad-event $B \in \mathcal B$ has $\Pr_{\Omega}(B) \leq p$ and is lopsidependent with at most other $d$ bad-events of $\mathcal B$ such that $e p d < 1$.

These two phases are almost completely distinct. The first is highly algebraic, while the second is more combinatorial. In this section, we will only discuss the first phase of constructing the generic set of bad-events to be compatible with the LLLL. The second phase, for which  we use only standard techniques,  is discussed in Appendix~\ref{app-kolm}.

\subsection{Framework for resampling oracles}
Consider a probability space $\Omega$ over a ground set $U$, along with a collection $\mathcal B$ of events in that space. There is also a binary symmetric relation $\sim$ provided for $\mathcal B$, which we refer to as the \emph{dependency relation}.\footnote{More properly, this should be referred to as a ``lopsidependency'' relation.  The distinction between dependency and lopsidependency is  not important for us so we use the simpler terminology.}   We will define the properties needed for a resampling oracle $\mathfrak R$ for this space, in the sense of Algorithm~\ref{gen-alg}, along with the new property ``obliviousness'' which we will need for our algorithms. We will later construct a number of such resampling oracles.

We will define $\mathfrak R$ by specifying a monoid $R$ which acts  on $U$. We refer to the $R$-act on $U$ as the \emph{resampling action}, and we write it as $r u$ for $r \in R, u \in U$.  We also define, for each $B \in \mathcal B$, a probability distribution $\Gamma_B$ over $R$ and we define $R_B = \text{Support}(\Gamma_B) \subseteq R$. The intent is to define the resampling oracle $\mathfrak R_B$ as $\mathfrak R_B(u) = r u$ where $r \approx \Gamma_B$.  Note that it is very important for us to separate the role of the randomness used in $\mathfrak R_B$.

Before we define our new obliviousness property, let us reiterate the conditions of Harvey \& Vondr\'{a}k \cite{harvey} and Kolmogorov \cite{kolmogorov}, in terms of our notation.\footnote{Kolmogorov \cite{kolmogorov} refers to property (C3) here as ``strong commutativity.'' We will never use the weaker commutativity properties defined by Kolmogorov, so we just refer to this as commutativity for convenience.} 
\begin{enumerate}
\item[(C1)] (Probability regeneration) For any $B \in \mathcal B$ and any fixed $v \in U$, we have
  $$
  \Pr_{(u, r) \approx (\Omega | B) \times \Gamma_B} ( r u = v) = \Omega[v] 
  $$
\item[(C2)] (Locality) If $B \not \sim B'$, and $u \in B - B'$, then for all $r \in R_B$ we have $r u \notin B'$.
  
\item[(C3)] (Commutativity)  Let $B_1 \not \sim B_2$. For any states $u \in B_1 \cap B_2$ and $u' \in U$, there is an injective mapping from states $w \in B_2 \cap R_1 u$ with $u' \in R_2 w$, to states $w' \in B_1 \cap R_2 u$ with $u' \in R_1 w'$, such that
  $$
  \Pr_{r_1 \approx \Gamma_{B_1}} (r_1 u = w) \Pr_{r_2 \approx \Gamma_{B_2}} (r_2 w = u')  = \Pr_{r_2 \approx \Gamma_{B_2}} (r_2 u = w') \Pr_{r_1 \approx \Gamma_{B_1}} (r_1 w' = u') 
$$
\end{enumerate}

\begin{observation}
If Properties (C1) and (C2) are satisfied, then the randomized function $\mathfrak R_B$ defined by choosing $r \approx \Gamma_B$ and outputting $\mathfrak R_B(u) = r u$, gives a resampling oracle in the sense of  Harvey \& Vondr\'{a}k \cite{harvey}. If (C3) is also satisfied, then the resampling oracle $\mathfrak R_B$ is commutative in the sense of Kolmogorov \cite{kolmogorov}.
  \end{observation}

We define a \emph{resampling-space} to be an ensemble of such objects $\mathcal B, R, U, \Omega, \sim$ satisfying (C1) and (C2). We sometimes  refer to the overall ensemble also just as $\mathcal B$.  We define the neighborhood of $B \in \mathcal B$ by $N(B) = \{ A \in \mathcal B: A \sim B \}$ and we also define $\overline N(B) = N(B) \cup \{B \}$.

Observe that if $\mathcal C, R, U, \Omega, \sim$ is a resampling-space and $\mathcal B \subseteq \mathcal C$, then $\mathcal B, R, U, \Omega, \sim$ is also a resampling-space (where $\sim$ is the restriction to $\mathcal B$). Furthermore, if (C3) holds for $\mathcal C$ then it holds for $\mathcal B$ as well.  We emphasize that these properties alone do not imply that that Algorithm~\ref{gen-alg} will converge when using the resampling oracle $\mathfrak R$. Our usual strategy is to show that some generic set $\mathcal C$ is a resampling-space with desired properties, and then take $\mathcal B$ to be an arbitrary subset of $\mathcal C$. We then show that one of the LLLL convergence criteria, such as Shearer's criterion, is satisfied on $\mathcal B$. See Appendix~\ref{app-kolm} for further details and definitions.

Bearing this in mind, we can summarize the main result of \cite{harvey} as follows:
\begin{theorem}[\cite{harvey}]
If $\mathcal B$ is a resampling-space which satisfies Shearer's criterion, then Algorithm~\ref{gen-alg} terminates in expected polynomial time. 
\end{theorem}

We are now ready to introduce the new structural property:
\begin{enumerate}
\item[(C4)] (Obliviousness) For all pairs $B, B'$ in $\mathcal B$ with $B \not \sim B'$, and all $r \in R_B$, one of the following two conditions holds:
  \begin{enumerate}
  \item For all $u \in B \cap B'$ we have $r u \in B'$
  \item For all $u \in B \cap B'$ we have $r u \notin B'$
    \end{enumerate}
\end{enumerate}

We refer to this as obliviousness since whether $r u$ is in $B$ does not depend upon the state $u$. In light of  (C4), let us define set $R_{B;B'} = \{ r \in R_B \mid  r u \in B'    \}$. We also define the conditional probability distribution $\Gamma_{B;B'} = \Gamma_{B} | R_{B;B'}$, and for any set $E \subseteq \mathcal B$ we define $R_{B;E} = \bigcap_{B' \in E} R_{B;B'}$ and $\Gamma_{B;E} = \Gamma_B | R_{B;E}$.

The definition of commutativity as it appears in (C3) is cumbersome to work with and lacks good compositional properties. To make it easier to show (C3), we use an additional property of resampling oracles identified by Achlioptas \& Iliopoulos \cite{achlioptas}, which we refer to as \emph{injectivity}.\footnote{In \cite{achlioptas}, this property is referred to as \emph{atomicity}. We use the alternate terminology \emph{injectivity} to avoid confusion with our discussion of atomic bad-events.} We state one variant of this property as follows:
\begin{enumerate}
\item[(C5)] (Injectivity) For all $u \in U$ and $B \in \mathcal B$, there is exactly one $w \in B$ with $u \in R_B w$.
\end{enumerate}

Our main motivation for this property is that it greatly simplifies condition (C3), allowing us to use an alternate condition (C3') instead:
\begin{enumerate}
\item[(C3')]  For all pairs  $B_1, B_2$ and all $u \in B_1 \cap B_2$ we have $R_{B_2} R_{B_1; B_2} u = R_{B_1} R_{B_2; B_1} u$.
\end{enumerate}
  
We summarize this in the following result:
\begin{proposition}
  \label{strong-atom-prop}
If properties (C3'), (C4), (C5) hold, then property (C3) holds.
\end{proposition}
\begin{proof}
  We begin with a  preliminary calculation: consider any $B \in \mathcal B, w \in B, u \in R_B w$. By (C1) we have $\Pr_{r \approx \Gamma_B, w' \approx \Omega|B} (r w' = u) = \Omega[u]$.   By (C5), we have $r w' = u$ only if $w' = w$, and so $\Pr_{r \approx \Gamma_B, w' \approx \Omega|B} (r w' = u) = \Pr_{r \approx \Gamma_B, w' \approx \Omega|B} (r w' = u \wedge w = w') = \Omega[w]/\Omega[B] \times \Pr_{r \approx \Gamma_B}(r w = u)$. Combining these equations, we get the following formula:
    \begin{equation}
    \label{gl5}
    \Pr_{r \approx \Gamma_B} (r w = u) = \Omega[u] \Omega[B] / \Omega[w].
    \end{equation}

  Let us now show (C3).  Fix $B_1, B_2, u, u'$. By (C5), at most one state $w$ has $w \in R_{B_1} u, u' \in R_{B_2} w$. If there is no such $w$, then there is nothing to show. Otherwise, by (C3') there must exist $w'$ with $u' \in R_{B_1} w', w' \in R_{B_2} u$.  We map $w$ to this $w'$. Since there is only one possible value $w$, the mapping is trivially injective. We need to show that this pair $w, w'$ satisfies
   $$
  \Pr_{r_1 \approx \Gamma_{B_1}} (r_1 u = w) \Pr_{r_2 \approx \Gamma_{B_2}} (r_2 w = u')  = \Pr_{r_2 \approx \Gamma_{B_2}} (r_2 u = w') \Pr_{r_1 \approx \Gamma_{B_1}} (r_1 w' = u') 
  $$

  By Eq.~(\ref{gl5}), we have
  $$
  \Pr_{r_1 \approx \Gamma_{B_1}} (r_1 u = w) \Pr_{r_2 \approx \Gamma_{B_2}} (r_2 w = u') = \frac{\Omega[u'] \Omega[B_1]}{\Omega[w]} \times \frac{\Omega[w] \Omega[B_2]}{\Omega[u]} = \frac{\Omega[u'] \Omega[B_1] \Omega[B_2]}{\Omega[u]}
  $$

  A symmetric argument shows that $\Pr_{r_2 \approx \Gamma_{B_2}} (r_2 u = w') \Pr_{r_1 \approx \Gamma_{B_1}} (r_1 w' = u')$ is also equal to this quantity.  
\end{proof}

\subsection{Atomically-generated probability spaces}
Most known resampling-spaces have a nicer form: the bad-events $B$ are conjunctions of a limited class of ``atomic'' events. For example, for the variable-assignment LLLL, an atomic event is $X_i = j$; for the space of uniform permutations, an atomic event is $\pi(x) = y$. The obliviousness property allows us to formalize this: we can define a resampling oracle and a simple list of axioms \emph{for the atomic events alone}, and then we automatically get a resampling oracle for \emph{conjunctions} of atomic events. This vastly simplifies the constructions for a number of diverse LLLL spaces.

Let $\mathcal A, R, U, \Omega, \sim$ be an oblivious resampling-space. We say that a set $E \subseteq \mathcal A$ is \emph{stable} if $A \not \sim A'$ for all distinct pairs $A, A' \in E$, and we define $\langle E \rangle = \bigcap_{A \in E} A$. For $A_1, \dots, A_k \in A$, we also write $\langle A_1, \dots, A_k \rangle$ as shorthand for $ A_1 \cap \dots \cap A_k = \langle \{ A_1, \dots, A_k \} \rangle$.

Let us define $\oa$ to be the set of conjunctions of events of $\mathcal A$,
$$
\oa = \bigl \{ \langle E \rangle \mid \text{$E$ a stable subset of $\mathcal A$} \bigr \}
$$
 We will use the same  ground set $U$ and monoid $R$ for $\oa$. The new dependency relation $\sim$ for $\oa$ is defined by setting $\langle E \rangle \sim \langle E' \rangle$ if there exist $A \in E, A' \in E'$ with $A \sim A'$. 

The key to the construction is to extend the distributions $\Gamma_{A}$ for the atomic events to a probability distribution $\Gamma_C$ for an event $C = \langle E \rangle$ in $\oa$. To do so,  we select some arbitrary fixed ordering as $E = \{A_1, \dots, A_k \}$, and we then define $\Gamma_C$ to be the distribution over products $r = r_k r_{k-1} \cdots r_2 r_1$, wherein $r_1, \dots, r_k$ are independent random variables and $r_i$ is drawn from distribution $\Gamma_{A_i; \{A_{i+1}, \dots, A_k \}}$.   (For $k = 0$, $r$ is the identity element of $R$.)

\begin{theorem}
  \label{atomic-thm}
  If $\mathcal A$ is an oblivious resampling-space, then so is $\oa$.  If, in addition, $\mathcal A$ satisfies (C5) and (C3'), then so does $\oa$; in particular, $\oa$ is commutative.
\end{theorem}

The proof of Theorem~\ref{atomic-thm} is technical, so we defer it to Appendix~\ref{atomic-thm-proof}. In later sections, we use it for a number of new and simpler constructions of resampling-spaces. Notably, these include  hamiltonian cycles of $K_n$ and perfect matchings of $K_n^{(s)}$. Our construction for hamiltonian cycles of $K_n$ is commutative, in contrast to a previous resampling oracle construction of Harvey \& Liaw \cite{harvey2}. No resampling oracle of any kind was known for perfect matchings of $K_n^{(s)}$ for any $s > 2$. 

\subsection{Efficient resampling oracles}
Our framework for resampling oracles, in which $\mathfrak R$ is derived from a monoid $R$, may seem overly restrictive. In fact, it is without loss of generality: for an arbitrary resampling oracle in the sense of Harvey \& Vondr\'{a}k \cite{harvey}, we could simply take $R$ to be the full transformation monoid. This would be useless computationally, because writing down an element of $R$ would require exponential time.

 In order to get an efficient parallel algorithm we must be able to efficiently compute on $R$. We summarize the requirements in terms of four properties (D0)---(D3); the runtime bounds are chosen so that the resampling action does not become the computational bottleneck for the overall algorithm described later. Here the parameter $n$ measures the input length to the algorithm.

\begin{enumerate}
\item[(D0)] We can sample from $\Omega$ in $O(\log^4 n)$ time and $\poly(n)$ processors.
\item[(D1)] For any $B \in \mathcal B$, we can sample from $\Gamma_B$ in $O(\log^3 n)$ time and $\poly(n)$ processors. 
\item[(D2)] For $r \in R$ and $u \in U$, we can compute $r u$ in $O(\log^3 n)$ time and $\poly(n)$ processors.
\item[(D3)] For $r, r' \in R$, we can compute $r r'$ in $O(\log^2 n)$ time and $\poly(n)$ processors.
\end{enumerate}

For atomically-generated probability spaces, these properties can themselves be simplified:
\begin{proposition}
  \label{simpb1b2a}
Suppose that $\mathcal B \subseteq \oa$, such that every bad-event $B \in \mathcal B$ is given by $B = \langle E \rangle$ for some stable set $E \subseteq \mathcal A$ with $|E| \leq \poly(n)$. Suppose that $\mathcal A$ satisfies property (D3) as well as the the following property (D1'):
  \begin{enumerate}
\item[(D1')] For any $A \in \mathcal A$ and stable set $E \subseteq \mathcal A$ with $A \not \sim E$ and $|E| \leq \poly(n)$, we can sample from $\Gamma_{A; E}$ in $O(\log^3 n)$ time and $\poly(n)$ processors.
  \end{enumerate}

Then $\mathcal B$ satisfies property (D1).
\end{proposition}
\begin{proof} 
Let $B = \langle E \rangle$ for some stable set $E = \{ A_1, \dots, A_k \}$ with $k \leq \poly(n)$. To draw $r \approx \Gamma_C$, we first use (D1') to sample independent variables $r_1, \dots, r_k$ wherein each $r_i$ drawn from $\Gamma_{A_i; \{A_{i+1}, \dots, A_k \}}$. We then use (D3) to compute $r = r_k \cdots r_1$ in $O(\log k \times \log^2 n) = O(\log^3 n)$ time.
\end{proof}

We say that a resampling space is \emph{amenable} if it satisfies the following computational conditions:
\begin{itemize}
\item It satisfies properties (C3)--(C4).
\item The monoid $R$ satisfies properties (D0)--(D3).
\item It has has a BEC running in $O(\log^3 n)$ time and $\poly(n)$ processors.
\end{itemize}

 We will later describe a parallel algorithm for such spaces. Note that, even without these properties, the resampling-space may still be be useful for a sequential algorithm or a combinatorial existence proof. Also, note that the third condition is satisfied if $m \leq \poly(n)$ and we can efficiently check each bad-event in $O(\log^3 n)$ time.

\subsection{Cartesian products}
Another useful method for constructing resampling-spaces comes from a cartesian product construction.   Consider resampling-spaces $\mathcal C_i, R_i, U_i, \Omega_i, \sim_i$ for $i = 1, \dots, s$.  We define a new resampling-space $\mathcal C = \mathcal C_1 \times \dots \times \mathcal C_s$ as follows. The underlying space is $U = U_1 \times \dots \times U_s$ and  $\Omega$ is the corresponding product distribution. The monoid $R$ is the cartesian product $R_1 \times \dots \times R_s$, with the natural monoid act on $U$. The events in $\mathcal C$ are those of the form $C_1 \times  \dots \times C_s$, where $C_i \in \mathcal C_i$. For such an event $C$, we define $\Gamma_{C}$ to be the probability distribution on tuples $(r_1, \dots, r_s)$, wherein $r_1, \dots, r_s$ are independent, and $r_i$ is drawn from $\Gamma_{C_i}$ in resampling-space $\mathcal C_i$. The relation $\sim$ on $\mathcal C$ is defined by $(C_1, \dots, C_s) \sim (C'_1, \dots, C'_s)$ if there is an index $i \in \{1, \dots, s \}$ where $C_i \sim_i C'_i$.

The following is immediate from the definitions:
\begin{observation}
  \label{cart-obs}
  If $\mathcal C_1, \dots, \mathcal C_s$ are oblivious resampling-spaces, then so is $\mathcal C$.

  If in addition $\mathcal C_1, \dots, \mathcal C_s$ are  commutative, then so is $\mathcal C$.
  
  If in addition $s \leq \poly(n)$ and $\mathcal C_1, \dots, \mathcal C_s$ satisfy properties (D0)--(D3), then so does $\mathcal C$.  
\end{observation}

As an example, the permutation LLL as defined in \cite{harris-srin-perm} allows selection of $s$ permutations $\pi_1, \dots, \pi_s$, wherein each $\pi_i$ is drawn independently and uniformly from some $S_{n_i}$, and a bad-event has the form $\pi_{i_1}(x_1) = y_1 \wedge \dots \wedge \pi_{i_k}(x_k) = y_k$. This can be modeled as the cartesian product of the uniform distributions on $S_{n_1}, \dots, S_{n_s}$. Therefore, the resampling action defined by the uniform distribution on $S_n$ immediately gives a corresponding resampling action for the permutation LLL.

\section{LFMIS for directed graphs}
\label{lfmis-sec}
Before we describe the parallel LLLL algorithm, we need an important graph-theoretic subroutine: the \emph{LFMIS for directed graphs}. This plays a similar role for our LLLL algorithm as the MIS does for the parallel MT algorithm. By itself, the LFMIS has little connection to the LLLL, and may be of independent combinatorial and algorithmic interest.

For an undirected graph $G$, an \emph{independent set} of $G$ is a vertex set $S$ where no two vertices in $S$ are adjacent in $G$. A \emph{maximal independent set} (MIS) has the additional property that no $T \supsetneq S$ is an independent set of $G$. There is a trivial sequential algorithm to find an MIS of $G$ by adding vertices one-by-one to $S$. The MIS produced by this sequential algorithm is referred to as the \emph{lexicographically first MIS (LFMIS)}.

With a slight abuse of terminology, we can extend the definition of LFMIS to a directed graph $G = (V,E)$.  Formally, we define the LFMIS of $G$ with respect to a permutation $\pi: [n] \rightarrow V$ to be the vertex set $I$ produced by the following sequential process:
\begin{algorithm}[H]
\centering
\begin{algorithmic}[1]
\State Initialize $I = \emptyset$ and $A = V$ \Comment{$A = $ alive vertices}
\For{$i = 1, \dots, n$}
\State \textbf{if} $\pi(i) \in A$ \textbf{then} Update $I \leftarrow I \cup \{ \pi(i) \}, A \leftarrow A - N^{\text{out}}(\pi(i))$
\EndFor
\end{algorithmic}
\caption{The sequential algorithm to find the LFMIS of directed graph $G$.}
\label{seq-lfmis1}
\end{algorithm}

An undirected graph $G$ can be viewed as a directed graph $G'$, where  every edge $(u,v) \in G$ corresponds to two directed edges $(u, v), (v, u) \in G'$. The LFMIS (in the usual sense) of $G$ is then identical to the directed LFMIS of $G'$.

The LFMIS problem for undirected graphs is P-complete in general \cite{lfmis-p}. However, Blelloch, Fineman, Shun  \cite{greedy2} described a simple parallel greedy algorithm to find the LFMIS of an undirected graph, when $\pi$ is chosen uniformly at random.  The algorithm can also be used for directed graphs. We summarize it as follows, where we define $P^{\pi}(v)$ for a vertex $v$ to be the set of vertices $w$ with $\pi^{-1}(w) < \pi^{-1}(v)$.
\begin{algorithm}[H]
\centering
\begin{algorithmic}[1]
\State Initialize $I = \emptyset$ and $A = V$
\While{$A \neq \emptyset$}
\State Let $J$ be the set of nodes $v \in A$ such that $A \cap N^{\text{in}}(v) \cap P^{\pi}(v) = \emptyset$
\State Update $I \leftarrow I \cup J, A \leftarrow A - J - \bigcup_{v \in J} N^{\text{out}}( v )$
\EndWhile
\end{algorithmic}
\caption{The parallel greedy algorithm to find the LFMIS of directed graph $G$.
\label{algo:greedylfmis}}
\end{algorithm}

This can be viewed as a parallel algorithm, where each iteration of identifying the residual source nodes $J$ and adding them to $I$, can be implemented in  $O(\log n)$ time and $O(m+n)$ processors. Alternatively, it can be viewed as a distributed algorithm, where each iteration requires $O(1)$ distributed communication rounds on $G$.  We get the following main result to analyze Algorithm~\ref{algo:greedylfmis}.

\begin{theorem}
  \label{greedy-th0}
  Algorithm~\ref{algo:greedylfmis} produces the LFMIS of $G$ with respect to $\pi$.  When $\pi$ is chosen uniformly at random, then Algorithm~\ref{algo:greedylfmis} terminates in $O(\log^2 n)$ rounds whp. In particular, Algorithm~\ref{algo:greedylfmis} runs in $\tilde O(\log^3 n)$ time on an EREW PRAM whp.
\end{theorem}

The analysis is very similar to the proof given in \cite{greedy2}, which showed that the (undirected) degrees are rapidly reduced when $G$ is an undirected graph. We defer the full proof of Theorem~\ref{greedy-th0} to Appendix~\ref{lfmis-app}, which shows a slightly stronger result. Note that Fischer \& Noever \cite{fischer} later showed that Algorithm~\ref{algo:greedylfmis} terminates in $O(\log n)$ rounds whp for undirected graphs. We conjecture that it should be possible to improve our analysis and show that Algorithm~\ref{algo:greedylfmis} runs in  $O(\log n)$ rounds whp on directed graphs as well.

\section{A generic parallel resampling algorithm}
\label{alg-sec}
We are now ready to describe our parallel algorithm for an amenable resampling-space. We recall that throughout, the parameter $n$ represents the description size of a configuration, such that a state $u$ is encoded in $\poly(n)$ bits. Correspondingly, our goal for an RNC algorithm is to achieve $\text{polylog}(n)$ runtime, $\poly(n)$ processors, and success probability $1 - n^{-\Omega(1)}$.

\begin{algorithm}[H]
\centering
\caption{The parallel LLLL algorithm}
\label{par-alg1}
\begin{algorithmic}[1]
\State Draw $u$ from the distribution $\Omega$
\While{there are true bad-events}:
\State Let $V$ denote the set of bad-events which are currently true
\State For each $B \in V$, independently draw a random variable $r_B \approx \Gamma_B$
\State Construct the directed graph $G$, whose vertex set is $V$, and whose directed edge set is

\vspace{-0.1in}
$$
E = \{ (B, B') \mid \text{$B \sim B'$ or $r_B \notin R_{B; B'}$} \}
$$
\State Find the LFMIS $I$ of $G$ with respect to a random permutation $\pi$
\State Sort $I = \{B_{i_1}, B_{i_2}, \dots, B_{i_s} \}$, where  $\pi(B_{i_1}) < \pi(B_{i_2}) < \dots < \pi(B_{i_s})$
\State Update the state as $
u \leftarrow r_{B_{i_s}} r_{B_{i_{s-1}}} \cdots r_{B_{i_1}} u
$
\EndWhile
\end{algorithmic}
\end{algorithm}

Clearly, if Algorithm~\ref{par-alg1} terminates, then all the bad-events in $\mathcal B$ are false on $u$. For maximum generality, we analyze Algorithm~\ref{par-alg1} in terms of two parameters $W, \epsilon$  from the Shearer LLLL criterion; see Appendix~\ref{app-kolm} for a precise definition. Theorem~\ref{equiv-thms} gives a few simpler LLL criteria, including the symmetric, asymmetric, and cluster-expansion criteria. For most applications, $W \leq \poly(n)$ and $\epsilon \geq \Omega(1)$. Our main result will be the following:
\begin{theorem}
  \label{par-bound2}
  Let $\mathcal B$ be an amenable resampling-space. If the Shearer criterion is satisfied with parameters $\epsilon, W$, then Algorithm~\ref{par-alg1} runs in $O( \frac{ \log^4(n + W \epsilon)}{\epsilon} )$   time and $\poly(n, W)$ processors whp.
\end{theorem}

For most applications, we can use a simplified corollary:
\begin{corollary}
  \label{par-bound4}
  Let $\mathcal B$ be an amenable resampling-space which satisfies the symmetric LLL criterion $ e p d (1 + \epsilon) < 1$.  Then Algorithm~\ref{par-alg1} runs in $O( \frac{\log^4(mn)}{\epsilon})$ time and $\poly(m, n)$ processors whp.
  \end{corollary}

Some probability spaces have convergence and distributional properties which go beyond the generic bounds such as Shearer's criterion \cite{harris-llll,harris-mt-dist,ilio}. Since Algorithm~\ref{par-alg1} can be viewed as a simulation of the sequential algorithm, all such bounds apply equally to it. We will see some examples in the next section with analysis of the variable-assignment LLLL.

We now turn to proving Theorem~\ref{par-bound2}. We assume throughout that $\mathcal B$ is amenable. We refer to each iteration of the main loop of Algorithm~\ref{par-alg1} (lines (3) -- (8)) as a \emph{round}. We use $V_t, I_t, \pi_t$, etc to denote the quantities corresponding to round $t$, and also define $b_t = |V_t|$.   We first observe that a single round can be implemented efficiently.
\begin{proposition}
  \label{run-prop1}
Each round of Algorithm~\ref{par-alg1} can be implemented using $\poly(b_t, n)$ processors and $O(\log^3 (b_t n))$ time whp. 
\end{proposition}
\begin{proof}
Since $\mathcal B$ is amenable, we can determine the set $V_t$ using our BEC in $O(\log^3 n)$ time. 

By  (D1), we can draw the random variables $r_B$ in time $O(\log b_t + \log^3 n)$.  In light of  (C4), we can efficiently check if $r_B \in R_{B;B'}$, by computing $r_B u$ and testing if $r_B u \in R_{B'}$.

By Theorem~\ref{greedy-th0}, we can find $I$ in time $O(\log^3(b_t n))$ and $\poly(b_t, n)$ processors whp. 

To implement step (7), we use use the associativity of monoid multiplication to compute the product $r_{B_{i_s}} \cdots r_{B_{i_1}}$ in $\lceil \log_2 s \rceil$ rounds of pairwise multiplications. By (D3), each round takes $O(\log^2 n)$ time. Noting that $s \leq b_t$, this gives a total of $O(\log^3 (b_t n))$ time and $\poly(b_t, n)$ processors. Once this product is computed, we can use (D2)  to compute $r_{B_{i_s}} \cdots r_{B_{i_1}} u$.
  \end{proof}

Thus, our main task is to show that Algorithm~\ref{par-alg1} terminates after a small number of rounds. We do so by coupling it to a sequential resampling algorithm, Algorithm~\ref{seq-alg1}.
\begin{algorithm}[H]
\centering
\caption{A sequential resampling algorithm}
\label{seq-alg1}
\begin{algorithmic}[1]
\State Draw $u$ from the distribution $\Omega$
\While{there are true bad-events} for $t = 1, 2, \dots $:
\State Let $V_t = $ set of bad-events true on $u$, and initialize $A = V_t$ \Comment{$A = $ alive events}
\State For each $B \in V_t$, draw a random variable $r_B \approx \Gamma_B$
\State Select a random ordering of $V_t$ as $V_t = \{B_1, \dots, B_k \}$.
\For{$k = 1, \dots, |V_t|$}
\If{$B_k \in A$}
\State Update $u \leftarrow r_{B_k} u$
\State \textbf{for}{ any $B' \in V_t$ with either (i) $B'$ is false on $u$ or (ii) $B' \sim B_k$} \textbf{do} $A \leftarrow A - \{ B' \}$
\EndIf
\EndFor
\EndWhile
\end{algorithmic}
\end{algorithm}

By the principle of deferred decisions, there is no difference in selecting the random variable $r_B$ in a ``preprocessed'' way (as in line (4) of Algorithm~\ref{seq-alg1}), as opposed to in ``online'' way as in Algorithm~\ref{gen-alg}. Thus, line (8) of Algorithm~\ref{seq-alg1} is equivalent to executing the resampling oracle $\mathfrak R_B(u)$ and so Algorithm~\ref{seq-alg1} can be viewed as a version of Kolmogorov's algorithm (Algorithm~\ref{kolm-alg}).

For Algorithm~\ref{seq-alg1},  define $\pi_t$ to be the chosen ordering of $V_t$, i.e. the map sending $i$ to $B_i$ in $V_t$. Also define $I'_t$ to be the set of events resampled in round $t$, i.e. the events $B_k$ such that $B_k \in A$ at iteration $k$ of line (7). The following result shows the equivalence between Algorithm~\ref{seq-alg1} and Algorithm~\ref{par-alg1}:
\begin{proposition}
\label{couple-prop2}
If the random variables $\pi, u, r$ are all fixed at the beginning of round $t$ and $I, I'$ are the LFMIS produced for Algorithms \ref{par-alg1} and \ref{seq-alg1} respectively for round $t$, then $I = I'$.
\end{proposition}
\begin{proof}
Let $u_j$ denote the state after iteration $j$ of round $t$ (and $u_0$ is the state at the beginning of round $t$). We have $V_t$ enumerated as $\{ B_1, B_2, \dots, B_k \}$ where $\pi(B_1) < \pi(B_2) < \dots < \pi(B_k)$, and we write $r_i$ as shorthand for $r_{B_i}$.

With this notation, observe that $B_j \in I$ iff there is no $i < j$ with $B_i \in I$ and either (a) $B_i \sim B_j$ or (b) $r_i \in R_{i} - R_{B_i;B_j}$. Similarly, $B_j \in I'$ iff there is no $i < j$ with $B_i \in I'$ and either (a) $B_i \sim B_j$ or (b) $B_j$ is false on $u_i$.  For contradiction, say that $j$ is minimal such that the membership of $B_j$ differs in $I$ and $I'$.

Suppose that $B_j \in I' - I$. Since $B_j \notin I$, there must be some $i < j$ with $B_i \in I$ such that $B_i \sim B_j$ or $r_i \notin R_{B_i; B_j}$. In the former case, by our induction hypothesis $B_i \in I'$ and this would contradict that $B_j \in I'$. In the latter case, note that since $B_j \in I'$, it must be that $B_j$ is true on $u_i$ and $u_{i-1}$ and $B_i$ is true on $u_{i-1}$. Thus, $u_{i-1} \in B_i \cap B_j$ and $u_i = r_{i} u_{i-1} \in B_j$. So  $r_i \in R_{B_i; B_j}$, a contradiction.

Next, suppose that $B_j \in I - I'$. Since $B_j \notin I'$, there must be some $i < j$ with $B_i \in I'$ such that $B_i \sim B_j$ or $B_j$ is false on $u_i$. Let $i$ be minimal subject to these conditions. In the former case, by induction hypothesis $B_i \in I$; in the latter case, by minimality of $i$, it must be that $B_j$ becomes false after resampling $B_i$, and so $B_i \in I$. In either case, we have $B_i \in I$. So $u_{i-1} \in B_i \cap B_j$ and $u_i = r_i u_{i-1} \notin B_j$, implying that $r_i \in R_{B_i} - R_{B_i; B_j}$. Thus $G$ has an edge $(B_i, B_j)$, contradicting that $B_j \in I$.
\end{proof}

They key property we need to analyze Algorithm~\ref{seq-alg1} is the following:
\begin{lemma}
\label{height-lemma}
If $B \in V_t$ for $t \geq 2$, then $\overline N(B) \cap I'_{t-1} \neq \emptyset$.
\end{lemma}
\begin{proof}
In the execution of Algorithm~\ref{seq-alg1}, let $T_i$ denote the total number of resamplings before round $i$ (so $T_1 = 0$), and note that $u^{T_i}$ is the state immediately at the beginning of round $i$.
  
By definition, $B$ must be true on $u^{T_t}$.  Either $B$ is true at time $T_{t-1}$ or $B \sim B' \in I'_{t-1}$; otherwise, by property (C4), $B$ would remain false after all the resamplings in round $t-1$. 

If $B \sim B' \in I'_{t-1}$ or $B \in I'_{t-1}$ we are done.  Otherwise, suppose $B \in V_{t-1} - I'_{t-1}$. This can only be the case if $B$ was marked as dead in round $t-1$. Suppose this occurs at time $i$, during the resampling of some $B' \in I'_{t-1}$. If $B \sim B'$, we are done.

Otherwise, suppose that $B$ is false on $u^{i}$. Since $B$ is true at the beginning of round $t$, by (C4) there must be some $B''$ resampled between times $i$ and $T_t$ with $B'' \sim B$, i.e. $B'' \in \overline N(B) \cap I'_{t-1}$.
\end{proof}

Lemma~\ref{height-lemma} in combination with analysis of Kolmogorov \cite{kolmogorov} shows that Algorithm~\ref{par-alg1} terminates in a small (polylogarithmic) number of rounds.  There is also a ``random-like'' distribution of the states during intermediate stages of the parallel LLLL algorithm. In all, we get the following bound:
\begin{lemma}
\label{length-bound}
Whp, Algorithm~\ref{par-alg1} terminates after $O(\frac{\log (n + W \epsilon)}{\epsilon})$ rounds and $\sum_t b_t \leq O(W \poly(n))$.
\end{lemma}

The proof of Lemma~\ref{length-bound} requires significant background and a number of preliminary definitions, so we defer it to Appendix~\ref{app-kolm}.

Now let $s = O(\frac{\log (n + W \epsilon)}{\epsilon})$ denote the total number of rounds in Algorithm~\ref{par-alg1}. Proposition~\ref{run-prop1} shows that each round $t$ uses $O(\log^3 (b_t n))$ time and $\poly(b_t, n)$ processors whp. Property (D0) allows us to implement step (1) in $O(\log^4 n)$ time. Thus the overall runtime of Algorithm~\ref{par-alg1} is at most $O(\log^4 n + s \log^3 n + \sum_{t=1}^{s} \log^3 (b_t))$. By concavity, we have $  \sum_{t=1}^{s} \log^3 b_t \leq s \log^3( 1 + \sum_{t=1}^{s} b_t/s )$. By Lemma~\ref{length-bound}, we have $\sum_t b_t \leq O(W \poly(n))$ whp. Thus, the time complexity here is as most $O(s \log^3( n + W \epsilon))$ and the processor count is at most $\poly(n,W)$.

This shows Theorem~\ref{par-bound2}.  Corollary~\ref{par-bound4} follows directly, noting that $W \leq O(m)$.

\section{The variable-assignment LLLL}
\label{var-llll-sec}
The variable-assignment LLLL is one of the most important LLLL probability spaces. Let us set notation and discuss how this fits into our resampling framework. We also discuss a few unique properties of the variable-assignment LLLL as well as some applications.

To begin the construction, we first consider the simplest setting, where the probability space $\Omega$ is defined by \emph{single} variable $X$ over a universe $U$. The generic bad-event set $\mathcal B$ has the tautological event $\top$, as well as an event $B_u \equiv X = u$ for each $u \in U$. We define $\sim$ by setting $B_u \sim B_{u'}$ for $u \neq u'$. (The event $\top$ is not dependent with any others.)

We form $R$ using a construction called the \emph{find-last monoid}. Formally, we define $R = U \cup \{ 1 \}$, where $1$ is an identity element. The binary operation on $R$ is defined as
$$
u' u = \begin{cases}
u' & \text{if $u' \neq 1$} \\
u & \text{if $u' = 1$}
\end{cases}
$$
Note that $U \subseteq R$, with $r u \in U$ for $u \in U$, and so $R$ naturally gives a left $R$-act on $U$.

For event $\top$, we define $\Gamma_{\top}$ to be the value $1$ with probability one.  For an event $B_u$, we define $\Gamma_{B_u}$ to be the distribution $\Omega$. One can easily verify that the resulting resampling oracle $\mathfrak R_{B_u}$ is defined by $\mathfrak R_{B_u}(x) = u'$, where $u'$ is drawn from the distribution $\Omega$, i.e. we resample the variable. It is trivial to verify that this resampling-space satisfies (C3'), (C4), (C5), and (D0)--(D3).

We can get the full variable-assignment LLLL via the cartesian product construction. Namely, the probability space is over $U = D^n$ for some discrete set $D$, and each bad-event $B$ has the form $(B_1, \dots, B_n)$, wherein $B_i$ is either $\top$ or an event  $X_i = j_i$. Equivalently, $B$ can be written as $B \equiv X_{i_1} = j_1 \wedge \dots \wedge X_{i_k} = j_k$. For such an event, we define $\Gamma_B$ as follows: For $r = (r_1, \dots, r_n) \approx \Gamma_B$, the entries $r_1, \dots, r_n$ are all independent, wherein  $r_i = 1$ for $B_i = \top$ and $r_i  \approx \Omega_i$ otherwise.  The resulting oracle $\mathfrak R_B$ is to simply resample the variables $X_{i_1}, \dots, X_{i_k}$. By Observation~\ref{cart-obs}, this resampling-space is again amenable.

This is a very notationally heavy way of describing a very simple probability space and a very simple resampling action. However, it illustrates how our resampling framework gives a non-trivial resampling-space (the full variable-assignment LLLL) by composing a few trivial building-blocks.

\subsection{Alternate LLLL criterion}
In \cite{harris-llll}, Harris described an alternative convergence criterion for the MT algorithm called \emph{orderability}.
 This is defined in terms of a function $\mu: \mathcal B \rightarrow [0,\infty)$; the full formal definitions are technical and are deferred to Appendix~\ref{alt-crit-app}. As our parallel algorithm for the variable-assignment LLLL can be viewed as an implementation of the MT algorithm, the orderability criterion can also be used to analyze Algorithm~\ref{par-alg1}. This gives the following result:
\begin{theorem}
\label{alt-lll-thm}
Let $\mu: \mathcal B \rightarrow [0,\infty)$ satisfy the orderability variable-assignment criterion with $\epsilon$-slack, and let $W = \sum_{B \in \mathcal B} \mu(B)$. If $\mathcal B$ has a BEC using $O(\log^4 n)$ time and $\poly(n)$ processors, then Algorithm~\ref{par-alg1} runs in $ O( \frac{ \log^4 (n + W \epsilon)}{\epsilon} )$   time and $\poly(n, W)$ processors whp.
\end{theorem}

As a example application, we get the following result:
\begin{proposition}
Suppose we have a $k$-SAT instance in $n$ variables, where each variable appears in at most $L =  \frac{2^{k+1} (1 - 1/k)^k}{(k-1)(1 + \epsilon)} - \frac{2}{k}$  clauses. Then there is a parallel algorithm to find a satisfying assignment in $O( \frac{\log^4 n}{\epsilon})$ time using $\poly(n)$ processors whp.
\end{proposition}
\begin{proof}
As shown in \cite[Theorem 4.1]{harris-llll}, the orderability criterion can be satisfied with slack $\epsilon$ satisfied under these conditions using the weighting function $\mu(B) = \frac{1 + \epsilon}{ (2 - 2/k)^k }$ for all $B \in \mathcal B$. Furthermore, $W \leq m \leq n k \leq n^2$ and we can implement a BEC by checking every clause.
\end{proof}

\subsection{Distributed algorithms}
The LOCAL model is a popular model for distributed graph algorithms. Here, in each round, a node in a graph can perform arbitrary computations and has unlimited communication with its neighbors. Distributed LLL algorithms can solve a number of graph problems in this setting, where each vertex $v$ has a set of associated bad-events $\mathcal B_v$ local to $v$, and bad-events in $\mathcal B_v$ and $\mathcal B_u$ are dependent iff the distance from $u$ to $v$ is bounded by some (problem-specific) constant.

As a simple example, consider finding a proper vertex-coloring. For each vertex $v$, we have some bad-events that $v$ chooses the same color as a neighbor $w$. Observe now that $\mathcal B_v$ and $\mathcal B_u$ are dependent iff there is some common vertex $w$, i.e. $\text{dist}(v,w) \leq 2$. See \cite{pettie} for a thorough discussion of this model of computation and applications to a number of graph-coloring problems.

Our parallel algorithm can be easily transformed into a distributed LLLL algorithm: 
\begin{proposition}
  \label{distrib-bound1}
  Suppose that the orderability variable-assignment criterion is satisfied with parameters $W,  \epsilon$. Then there is a distributed LOCAL algorithm to find a variable assignment avoiding $\mathcal B$  in $O( \frac{\log^3 (Wn)}{\epsilon})$ rounds whp. In particular, if $e p d (1+\epsilon) \leq 1$, then this runs in $O( \frac{\log^3(mn)}{\epsilon})$ rounds.
\end{proposition}
\begin{proof}
All of the steps in a round $t$ of Algorithm~\ref{par-alg1}, except the computation of the LFMIS at line (6) and the state update at line (8) can be implemented in $O(1)$ communication rounds. The state update can be done in $O(\log(b_t n))$ rounds and the greedy LFMIS  can be implemented in $O(\log^2 (b_t n))$ rounds whp; note that Algorithm~\ref{par-alg1} only creates an edge between $B, B'$ if $B, B'$ overlap on a variable and so we can simulate the directed graph created in line (5). As shown in Appendix~\ref{alt-crit-app} we have  $b_t \leq W \poly(n)$ whp.
\end{proof}

One application, which is an immediate consequence of LLLL analysis of \cite{harris-llll},  is to proper vertex coloring of a hypergraph:
\begin{proposition}
Let $H$ be a $k$-uniform hypergraph in which each vertex appears in at most $L$ edges. Then there is a randomized LOCAL algorithm in $O( \frac{\log^3 n}{\epsilon})$ rounds to construct a non-monochromatic $c$-coloring of $H$ for $L \leq \frac{c^k (1-1/k)^{k-1}}{k (c-1)(1+\epsilon)}$.
\end{proposition}

\section{Other resampling-spaces}
\label{other-llll-sec}
We now discuss how our resampling framework applies to a few other resampling-spaces, with some applications. The main space discussed here is the uniform distribution on $S_n$. Two others are the uniform distribution on hamiltonian cycles of $K_n$, and the uniform distribution on perfect matchings of the complete hypergraph $K_n^{(s)}$ for $s \geq 2$. The latter two involve very technical algebraic arguments, so we defer the full proofs to Appendices~\ref{hamilton-sec} and~\ref{match-sec2}.

\subsection{Uniform distribution on $S_n$}
\label{perm-lll-sec}
In this setting, we have $U = S_n$, and we use $\pi$ instead of $u$ to represent the system state.  The atomic sets have the form
$$
A = \{ \pi \in S_n \mid \pi(x) = y \}
$$
for some $(x,y) \in [n] \times [n]$; we write this as $A = \langle (x, y) \rangle$. We define $\sim$ on $\mathcal A$ by setting $\langle (x,y) \rangle \sim \langle (x',y') \rangle$ if one of the following two conditions holds: (i) $x = x'$ and $y \neq y'$ or (ii) $x \neq x'$ and $y = y'$. Equivalently, we have $A \sim A'$ iff $\Pr_{\Omega}(A \cap A') = 0$.

We define $R$ to be the symmetric group $S_n$. For any $A  = \langle(x, y)\rangle$, we define $R_A$ to be the set of single-swap permutations of the form $\sigma = (y \ z)$ for $z \in [n]$, and $\Gamma_A$ is the uniform distribution on $R_A$. We define the resampling action as left-multiplication in the obvious way.

\begin{proposition}
  Properties (D0),  (D2) and (D3) hold.
\end{proposition}
\begin{proof}
The monoid operation and monoid act are both composition of permutations, which can easily be done in $O(\log n)$ time. Property (D0) holds using any of the standard ways to generate uniform random permutations.
 \end{proof}

\begin{proposition}
Properties (C5) and (C1) hold.
\end{proposition}
\begin{proof}
  Consider $A = \langle (x, y ) \rangle$ and $\pi \in S_n$. We claim that there is precisely one pair $(z, \tau)$ with $z \in [n], \tau \in A$ such that $\pi = (y \ z) \tau$. For, we have $(y \ z) \pi \in A$ iff $(y \ z) \pi x = y$ iff $z = \pi x$. Furthermore, once $z$ is determined, $\tau$ is also uniquely determined.

  This shows (C5). Also, when $z \approx \text{Unif}[n]$ and $\tau \approx \Omega|A$, it implies that $(y \ z) \tau = \pi$ with probability precisely $\frac{1}{n} \times \frac{1}{|A|} = \frac{1}{n!}$. Thus $(y \ z) \tau$ is uniformly distributed on $S_n$, showing (C1).
\end{proof}

\begin{proposition}
Property (C2) holds.
\end{proposition}
\begin{proof}
  Consider $A = \langle (x, y) \rangle$ and $A' = \langle (x', y') \rangle$ and $\pi \in A - A'$. Clearly $A \neq A'$ so $x' \neq x, y' \neq y$. Suppose for contradiction that $(y \ z) \pi x' = y'$. So $\pi x' = (y \ z) y'$. If $z \neq y'$, then $\pi x' = y'$, which contradicts $\pi \notin A'$. If $z = y'$, then $\pi x' = y$, which is impossible as $\pi x = y$.
\end{proof}

\begin{proposition}
\label{c6x-prop-11a}
Let $A = \langle (x, y ) \rangle$ and $A' = \langle (x', y') \rangle$ with $A \not \sim A'$. Let $\sigma = (y \ z) \in R_A$ and $\pi \in A \cap A'$. Then:
\begin{enumerate}
\item If $(x,y) = (x', y')$, then $\sigma \pi \in A'  \Leftrightarrow z = y$;
\item If $(x,y) \neq (x', y')$, then $\sigma \pi \in A' \Leftrightarrow z \neq y'$
\end{enumerate}
\end{proposition}
\begin{proof}
  In case (1), if $z = y$, then $\sigma \pi = \pi$, which is in $A = A'$ by hypothesis. If $z \neq y$, then $\sigma \pi x = (y \ z) y = z \neq y$, and so $\sigma \pi \not \in A$.

  In case (2), since $A \not \sim A'$ we have $y \neq y'$. If $z \neq y'$, then $\sigma \pi x' = (y \ z) y' = y'$ and so $\sigma \pi \in A'$. If $z = y'$, then  $\sigma \pi x' = (y \ y') y' = y \neq y'$, and so $\sigma \pi \notin A'$.
\end{proof}

\begin{proposition}
Property (C4) holds.
\end{proposition}
\begin{proof}
  Proposition~\ref{c6x-prop-11a} gives an explicit condition for when $\sigma \pi \in A'$ holds for $A \not \sim A', \pi \in A \cap A', \sigma \in R_A$. This condition depends solely on $A, A', \sigma$ and not on $\pi$ itself; thus, for any fixed $\sigma$, it holds for all such $\pi$ or none of them.
\end{proof}

\begin{proposition}
Property (C3') holds.
\end{proposition}
\begin{proof}
  Let $A_1 = \langle (x_1,y_1) \rangle$ and $A_2 = \langle (x_2, y_2) \rangle$ where $A_1 \not \sim A_2$. We need to show for any fixed $\pi$ and indices $z_1 \in [n] - \{y_2 \}, z_2 \in [n]$, there exist $z_1' \in [n], z_2' \in [n] - \{y_1 \}$ such that
  $$
  (y_2 \ z_2) (y_1 \ z_1) \pi = (y_1 \ z_1') (y_2 \ z_2')  \pi
  $$

If $A_1 = A_2$ this is trivial. Also, if $z_1, z_2$ are distinct from each other and $y_1, y_2$, then we can simply take $z'_1 = z_1, z_2' = z_2$. Otherwise, there are a number of cases depending on which of the terms $z_1, z_2, y_1, y_2$ are equal to each other. 

\textbf{Case I: $\pmb{z_1 = z_2}$}. Let $z = z_1 = z_2$.  If $z = y_1$, then $(y_2 \ z_2) (y_1 \ z_1) = (y_2 \ y_1)(y_1 \ y_1) = (y_1 \ y_2) (y_2 \ y_2)$,  and so setting $z'_1 = y_2, z_2' = y_2$ works. Otherwise, if $z \neq y_1$, then $(y_2 \ z_2) (y_1 \ z_1) = (y_1 \ y_2 \ z) = (y_1 \ y_2) (y_2 \ z)$. So setting $z'_2 = z, z'_1 = y_2$ works. Our hypothesis $z \neq y_1$ ensures that $z'_2 \neq y_1$.

\textbf{Case II: $\pmb{z_2 = y_2}$}. Then $(y_2 \ z_2) (y_1 \ z_1) = (y_1 \ z_1) = (y_1 \ z_1) (y_2 \ y_2)$, so take $z'_1 = y_1, z'_2 = y_2$.

\textbf{Case III:  $\pmb{z_2 = y_1}$}. We may assume that $z_2 \notin \{ z_1, y_2 \}$, as we have already covered these cases. Then $(y_2 \ z_2) (y_1 \ z_1)  =  (y_1 \ z_1 \ y_2) = (y_1 \ z_1) (y_2 \ z_1)$,
so taking $z'_2 = z'_1 = z_1$ works. Note that $z_2' \neq y_1$, as otherwise we would have $z_1 = z_2$.

\textbf{Case IV: $\pmb{z_1 = y_1}$}. Then $(y_2 \ z_2) (y_1 \ z_1) =  (y_1 \ y_1) (y_2 \ z_2)$, so take $z'_1 = y_1, z'_2 = y_2$. Note that we cannot have $z'_2 = y_1$ as this would imply $y_1 = y_2$.
\end{proof}

\subsection{Applications}
\label{app-sec2}
We illustrate with the classic applications of the permutation LLL to Latin transversals. Suppose we have an $n \times n$ matrix $A$, whose entries come from some set of colors. An \emph{$s$-bounded transversal} of this matrix is a permutation $\pi \in S_n$, such that no color appears at least $s$ times among the entries $A(i, \pi(i))$. The case $s = 2$ is known as a \emph{Latin transversal}, and in this case the permutation is said to be \emph{rainbow} in that no color is repeated among the entries of $A(i, \pi(i))$.

\begin{proposition}
  \label{ggt1}
  Suppose that each color appears at most $\Delta$ times in $A$. Then, we can find a Latin transversal $\pi \in S_n$ in $O(\log^4 n)$ time and $\poly(n)$ processors for $\Delta \leq 0.105 n$. We can find an $s$-bounded transversal $\pi \in S_n$ in $O(\frac{\log^4 n}{\epsilon})$ time and $\poly(n)$ processors for $\Delta \leq n \Bigl( \frac{(s-1)!}{2 e (1+\epsilon) s} \Bigr)^{1/(s-1)}$.
  \end{proposition}
  \begin{proof}
We use the probability space of the uniform distribution over $S_n$.   For the first result, observe that the cluster-expansion LLL criterion is satisfied with slack of $\epsilon = \Omega(1)$ and $W \leq \poly(n)$.
    
    For the second result, for each tuple $t = \{ (i_1, j_1), \dots, (i_s, j_s) \}$ with $A(i_1,j_1) = \dots = A(i_s,j_s)$, we have a separate bad-event $B_t$, that $\pi(i_1) = j_1 \wedge \dots \wedge \pi(i_s) = j_s$. Each $B_t$ has probability $p \leq \frac{(n-s)!}{n!}$, and has at most $d = 2 s n \binom{\Delta-1}{s-1}$ neighboring bad-events $B_{t'}$. Thus, in order to satisfy the symmetric LLL criterion with $\epsilon$-slack, we need
  $$
  e (1+\epsilon) \frac{(n-s)!}{n!} 2 s n \binom{\Delta-1}{s-1} \leq 1
  $$

To show this, we calculate:
  \begin{align*}
    e (1+\epsilon) \frac{(n-s)!}{n!} 2 s \binom{\Delta-1}{s-1} \leq \frac{2 e (1+\epsilon) s \Delta  \dots \Delta}{(s-1)! n \times n \dots n} = \frac{2 e (1+\epsilon) s \Delta^{s-1}}{(s-1)! n^{s-1}}
  \end{align*}

  So $e p d (1+\epsilon) \leq 1$ holds under the stated hypothesis. One can easily construct a BEC in $O(\log n)$ time: for each color class, simply enumerate all of the current entries of $\pi$ with that color.
  \end{proof}
 
We note that the runtime in Proposition~\ref{ggt1} does not depend on $s$. By contrast, the permutation LLL algorithm of \cite{harris-srin-perm} would only give a parallel algorithm for \emph{constant} $s$. There are two main reasons it has poor scaling as a function of $s$: first, the number of bad-events could be $n^s$, which is super-polynomial for unbounded $s$; second, each bad-event spans $s$ entries, whereas \cite{harris-srin-perm} only allows bad-events to use polylogarithmic entries. We also note that a sequential algorithm of \cite{harris-srin-perm} based on partial resampling  can achieve better bounds for large $s$, but our parallelization strategy does not extend to that case.

We next illustrate with some applications to finding rainbow subgraphs of $K_n$ and $K_n^{(s)}$:
\begin{proposition}
  Consider an edge-coloring of $K_n$ where every color appears on at most $\Delta$ edges. If $\Delta \leq 0.105 n$ and $n$ is even, then we can find a rainbow perfect matching in $O(\log^4 n)$ time and $\text{poly}(n)$ processors whp. If $\Delta \leq 0.026 n$, then we can find a rainbow hamiltonian cycle in $O(\log^4 n)$ time and $\text{poly}(n)$ processors whp.
\end{proposition}
\begin{proof}
  We encode these problems via the probability spaces of the uniform distribution of perfect matchings of $K_n$ and hamiltonian cycles of $K_n$, respectively.  In Apppendices~\ref{hamilton-sec} and~\ref{match-sec2} we show that the spaces both have amenable resampling oracles.  It is shown in \cite{kolmogorov} and \cite{harvey2}, respectively, that that cluster-expansion LLL criterion is satisfied with slack $\epsilon = O(1)$ and $W \leq \text{poly}(n)$.
  \end{proof}

\begin{proposition}
Consider an edge-coloring of $K_n^{(s)}$ where every color appears on at most $\Delta$ edges. If $\Delta \leq \frac{ \binom{n-s-1}{s-1} (1 - \frac{1}{2 s})^{2 s}}{ 2 s - 1 } n$, then there is a poly-time algorithm to find a rainbow perfect matching.
\end{proposition}
\begin{proof}
The probability space $\Omega$ is defined by selecting matching $M$ uniformly at random. For each pair of edges $e, e'$ of the same color, we have a bad-event $B$ that $e, e'$ are both in $M$. This event has probability 
  $$
  p = \frac{(n/s) (n/s - 1)}{\binom{n}{s} \binom{n-s}{s}}.
  $$

  In Appendix~\ref{match-sec2}, we show that $\Omega$ has a resampling-space, albeit not a commutative one. To apply the cluster-expansion criterion, we use a slightly denser dependency graph: two events $B, B'$ are dependent if the corresponding edges overlap. To enumerate the stable sets of neighbors of $B$ with respect to this dependency graph, for each of the $2s$ vertices $j$ involved in $B$ we may select another edge $f_j \ni j$ and another edge $f'_j$ of the same color as $f_j$ (a total of $\binom{n-1}{s-1} \times (\Delta-1)$ choices).

  We set $\mu(B) = \alpha$ for every bad-event for some parameter $\alpha \geq 0$. In order to satisfy the cluster-expansion criterion, we will then need
  \begin{equation}
  \label{gtr1}
  \alpha \geq \frac{(n/s) (n/s - 1)}{\binom{n}{s} \binom{n-r}{s}} \times (1 + \tbinom{n-1}{s-1} (\Delta - 1) \alpha)^{2 s}
  \end{equation}

Simple calculus shows that when the hypotheses are satisfied, then Eq.~(\ref{gtr1}) can be satisfied for some $\alpha \geq 0$. Using the resampling oracle in Appendix~\ref{match-sec2}, we can implement Algorithm~\ref{gen-alg} in polynomial time to produce a configuration avoiding $\mathcal B$.
\end{proof}

As another application, consider \emph{strong coloring}: given a graph $G$ with a partition of the vertices into $k$ blocks each of size $b$ (i.e., $V = V_1 \sqcup \dots \sqcup V_k$), we would like to find a proper $b$-coloring such that every block has exactly $b$ colors.   In \cite{haxell2008}, Haxell showed that such a coloring exists when $b \geq c\Delta$ and $\Delta$ is sufficiently large, for some constant $c \leq 11/4$; this is the best bound currently known. Furthermore, the constant $11/4$ cannot be improved to any number strictly less than $2$. In \cite{harris-srin-perm}, a variety of LLL-based algorithms are given for constructing the colorings, with worse bounds on $b$ and with large (unspecified) runtimes. Our LLLL algorithms gives a crisp result, which is perhaps the first parallel algorithm with reasonable bounds on both $b$ and the run-time:
\begin{proposition}
\label{strong-chrom-prop}
Given a partition of $G$ into blocks of size $b \geq (\frac{256}{27} + \epsilon) \Delta$, a coloring of $G$ can be found in $O(  \frac{\log^4 n}{\epsilon})$ time whp.
\end{proposition}
\begin{proof}
Consider the probability space of uniform distribution over permutations $\pi_1, \dots, \pi_k$, wherein each $\pi_i$ is a permutation of the vertices in block $V_i$. For each edge $e = (v, v')$ with $v \in V_i, v' \in V_{i'}$, and each value $\ell = 1, \dots, b$, we have a bad-event $\pi_{i}(v) = \ell \wedge \pi_{i'}(v') = \ell$. Harris \& Srinivasan \cite{harris-srin-perm} show that this satisfies the LLLL cluster-expansion criterion with $\epsilon$-slack when $b \geq (\frac{256}{27}  + \epsilon) \Delta$. Furthermore, the probability space is the cartesian product of $k$ copies of the uniform distribution on $S_b$.  By Observation~\ref{cart-obs}, this has an amenable resampling-space.
\end{proof}

We note that, subsequent to the original version of this paper, a variety of works have appeared with better bounds on the colors and the runtime for strong coloring \cite{graf, graf2, harris-det}. Most recently, \cite{harris-det} provides a deterministic sequential poly-time algorithm for $b \geq (3 + \epsilon) \Delta$ and a deterministic parallel algorithm with $O(\log^3 n)$ runtime for $b \geq (5+ \epsilon) \Delta$, for any constant $\epsilon > 0$.

Finally, we consider a hypergraph packing problem of Lu \& Sz\'{e}k\'{e}ly \cite{random-inj}. 

\begin{proposition}
  \label{ppr1}
  Let $H_1, H_2$ be two $s$-uniform hypergraphs on $n$ vertices, where each $H_i$ has $m_i$ edges such that $(d_1 + 1) m_2 + (d_2 + 1) m_1 < \frac{\binom{n}{s}}{e (1+\epsilon)}$.

  There is an algorithm in $\poly(n)$ processors and $\tilde O(\frac{\log^4 n}{\epsilon})$ time to find an injective map $\phi: V(H_2) \rightarrow V(H_1)$ such that $\phi(H_2)$ is edge-disjoint to $H_1$. (That is, there are not edges $f_1 \in H_1, f_2 \in H_2$ with $f_1 = \{ \phi(v) \mid v \in f_2 \}$.)
\end{proposition}
\begin{proof}
Let us briefly review a construction of \cite{random-inj}. We use the LLL to construct the permutation $\phi$. For each pair of edges $f_1 = \{ u_1, \dots, u_s \} \in E(H_1),  f_2 = \{ v_1, \dots, v_s \} \in E(H_2)$, and each permutation $\sigma \in S_s$, we form a bad-event that $\phi(v_1) = u_{\sigma 1} \wedge \dots \wedge \phi(v_r) = u_{\sigma s}$.  The stated hypothesis ensures that these events satisfy the symmetric LLL criterion. Furthermore, there is a simple BEC here which can be implemented in $O(\log n)$ time: for each $f_2$, we sort $\phi(f_2)$ and check if it in $H_1$.
\end{proof}

Note that Harris \& Srinivasan \cite{harris-srin-perm} only gives an RNC algorithm if the hypergraphs $H_i$ have rank $\text{polylog}(n)$; this condition is not required for Proposition~\ref{ppr1}.

\section{Acknowledgments}
Thanks to Chen Meiri for explanations about group actions. Thanks to anonymous conference and journal reviewers for helpful suggestions and corrections.

\appendix

\section{Background on the LLLL}
\label{app-kolm}
Consider some resampling-space $\mathcal B$ with a lopsidependency relation $\sim$. The simplest criterion for the LLL or the LLLL on $\mathcal B$ is the symmetric criterion $e p d \leq 1$, where $p$ is the maximum probability of any bad-event and $d$ is the maximum dependency of any bad-event. A number of other criteria such as the asymmetric criterion can also be stated in terms of the probabilities and dependency-structure of the bad-events; the most general of these is \emph{Shearer's criterion} \cite{shearer}.  Parallel algorithms usually need a slightly stronger criterion which we refer to as \emph{$\epsilon$-slack}:  the vector of probabilities $(1+\epsilon) \Pr_{\Omega}(B)$ must satisfy Shearer's criterion for  $\epsilon > 0$.

We will describe the Shearer criterion in terms of stable-set sequences, which is a more useful tool for analyzing the MT algorithms. The connection between stable-set sequences and the original form of Shearer's criterion was developed by Kolipaka \& Szegedy \cite{kolipaka}.

We say that a set $J \subseteq \mathcal B$ is \emph{stable} if there are not distinct elements $B, B' \in J$ with $B \sim B'$.  For a stable set $J$, we define $\overline N(J) = \bigcup_{B \in J} \overline N(B)$.

We define a \emph{stable-set sequence} to be a sequence $S = (S_1, S_2, \dots, S_{\ell})$, where each $S_i$ is a non-empty stable set of $\mathcal B$ and $S_i \subseteq \overline N(S_{i+1})$ for $i=1, \dots, \ell-1$. We say that $S$ is \emph{singleton} and \emph{rooted at $B$} if $S_\ell = \{ B \}$. We define the \emph{depth} of $S$ to be $\ell$, the \emph{size} of $S$ to be $|S| = \sum_{i=1}^{\ell} |S_i|$ and the \emph{weight} of $S$ to be $w(S) = \prod_{i=1}^{\ell} \prod_{B \in S_i} \Pr_{\Omega}(B)$. We define $\mathfrak S$ to be the set of all singleton stable-set sequences. 

\begin{theorem}[\cite{kolipaka}]
  \label{kolipaka-thm1}
If Shearer's criterion is satisfied with $\epsilon$-slack, then $\sum_{S \in \mathfrak S}(1+\epsilon)^{|S|}  w(S) < \infty$.
\end{theorem}

In light of Theorem~\ref{kolipaka-thm1}, we define the key parameter  $W = \sum_{S \in \mathfrak S} (1 + \epsilon)^{|S|} w(S)$. This allow us to state the most general bounds.  However, Shearer's criterion is difficult to work with in practice, so a number of simpler LLL criteria are often used instead. 

\begin{theorem}
  \label{equiv-thms}
1. (Asymmetric LLL criterion) Suppose that some function $x: \mathcal B \rightarrow [0, 1]$ satisfies
$$
\forall B \in \mathcal B \qquad \Pr_{\Omega}(B) (1 + \epsilon) \leq x(B) \prod_{A \in N(B)} (1 - x(A))
$$
Then Shearer's criterion is satisfied with $\epsilon$-slack and $W \leq \sum_{B \in \mathcal B} \frac{x(B)}{1-x(B)}$.

2. (Cluster-expansion criterion \cite{bissacot}) Suppose that some function $\mu: \mathcal B \rightarrow [0,\infty)$ satisfies
$$
\forall B \in \mathcal B \qquad \mu(B) \geq \Pr_{\Omega}(B) (1+\epsilon) \sum_{\substack{I \subseteq \overline N(B)\\\text{$I$ stable}}} \prod_{A \in I} \mu(A)
$$
Then Shearer's criterion is satisfied with $\epsilon$-slack and $W \leq \sum_{B \in \mathcal B} \mu(B)$.

3. (Symmetric LLL criterion) Suppose that $\Pr_{\Omega}(B) \leq p$ and $|\overline N(B)| \leq d$ for every $B \in \mathcal B$, and $e p d (1+\epsilon) \leq 1$. Then Shearer's criterion is satisfied with $\epsilon$-slack and $W \leq e m p$.
\end{theorem}

For each bad-event $B \in V_i$ during Algorithm~\ref{seq-alg1}, we define a corresponding sequence $\hat S(B,i) = (S_1, \dots, S_i)$ by setting $S_i = \{ B \}$ and then, for $j = i-1, \dots, 1$, setting $S_j = I'_j \cap \overline N( S_{j+1} )$. 

\begin{proposition}
  For $B \in V_i$, the sequence $\hat S(B,i)$ is a stable-set sequence of depth $i$ rooted at $B$.
\end{proposition}
\begin{proof}
  Clearly $\hat S(B,i)$ has depth $i$ and $S_{i} = \{ B \}$, and also clearly $S_j \subseteq \overline N(S_{j+1})$. Since $I'_j$ is stable, so is each $S_j$. Finally, to show that $S_j$ is non-empty, consider some $A \in S_{j+1}$; note that Lemma~\ref{height-lemma} ensures that there is some $A' \in \overline N(A) \cap I'_{j-1}$; this $A'$ will appear in $S_j$. 
\end{proof}

We say that a given depth-$i$ stable-set sequence $S$ rooted at $B$ \emph{appears} if $\hat S(B,i) = S$. Iliopoulos \cite{ilio} showed a connection between appearing stable-sequences and probabilities of bad-events in Algorithm~\ref{gen-alg} for a commutative resampling oracle. These bounds also apply to Algorithm~\ref{seq-alg1} since it is a version of Algorithm~\ref{gen-alg}. We summarize the key result as follows:

\begin{proposition}[\cite{ilio}]
\label{k-prop0}
For a commutative resampling oracle, any stable-set sequence $S$ appears with probability at most $w(S)$.
\end{proposition}

Using our bounds on stable-set sequences and arguments from \cite{hh}, we now prove Lemma~\ref{length-bound}:
\begin{proof}[Proof of Lemma~\ref{length-bound}]
Each $B \in V_i$ corresponds to an appearing depth-$i$ stable-set sequence $\hat S(B,i)$. All such stable-set sequences are distinct: if $i \neq i'$, then the depths of $\hat S(B,i)$ and $\hat S(B,i')$ are distinct, while if $B \neq B'$ then the roots of $\hat S(B,i)$ and $\hat S(B', i)$ are distinct.

  Thus, $\sum_i |V_i|$ is at most the number of appearing stable-set sequences. Proposition~\ref{k-prop0} shows that $\bE[ |V_i| ] \leq \sum_{S \in \mathfrak S}  w(S) \leq W$. So by Markov's inequality, $\sum_i |V_i| \leq \text{poly}(n) W$ whp.

 If Algorithm~\ref{seq-alg1} runs for $t$ rounds, then for each $i = 1, \dots, t$, there is at least one appearing depth-$i$ stable set sequence (namely $\hat S(B,i)$ for an arbitrary $B \in V_i$).  Thus, a necessary condition for Algorithm~\ref{seq-alg1} to run for $t$ rounds is that at least $t/2$ distinct singleton stable-set sequences of size at least $t/2$ appear. By Proposition~\ref{k-prop0}, the expected number of such sequences is given by
  {\allowdisplaybreaks
  \begin{align*}
  \sum_{\substack{S \in \mathfrak S \\ |S| \geq t/2}} w(S) =   (1 + \epsilon)^{-t/2} \sum_{\substack{S \in \mathfrak S \\ |S| \geq t/2}} w(S) (1+\epsilon)^{t/2} \leq  (1 + \epsilon)^{-t/2} \sum_{S \in \mathfrak S} w(S) (1+\epsilon)^{|S|} = (1 + \epsilon)^{-t/2} W
  \end{align*}
}
  By Markov's inequality, the probability that the actual number exceeds $t/2$ is at most $\frac{ (1 + \epsilon)^{-t/2} W }{t/2}$. This is below $n^{-\Omega(1)}$ for some $t = \Theta( \log( n +W \epsilon)/\epsilon)$.
\end{proof}

\section{Alternative variable-assignment LLLL criterion}
\label{alt-crit-app}
We summarize here an alternate criterion of Harris for the variable-assignment LLLL \cite{harris-llll}.

Given a bad-event $B$ of the variable-assignment LLLL and a set $E \subseteq \overline N(B)$, we say that $E$ is \emph{orderable} to $B$ if either $E = \{B \}$, or there is an ordering $B \equiv X_{i_1} = j_1 \wedge \dots \wedge X_{i_k} = j_k$ and an ordering $E = \{B_1, \dots, B_{k'} \}$ such that, for each $\ell = 1, \dots, k$, the bad-event $B_{\ell}$ demands $X_{i_{\ell}} \neq j_{\ell}$ and none of the events $B_{\ell'}$ for $\ell' < \ell$ do so. We also say that a map $\mu: \mathcal B \rightarrow [0,\infty)$ satisfies the orderability criterion with $\epsilon$-slack for $\mathcal B$ if it satisfies
$$
\forall B \in \mathcal B \qquad \mu(B) \geq (1+\epsilon) \Pr_{\Omega}(B) \sum_{\substack{\text{$E \subseteq \overline N(B)$} \\ \text{$E$ orderable to $B$}}} \prod_{A \in E} \mu(A)
$$

The main result of \cite{harris-llll} is the following:
\begin{theorem}
Suppose that the map $\mu$ satisfies the orderability criterion with $\epsilon$-slack for $\mathcal B$. Then the expected number of resampling executed by the MT algorithm  is at most $\sum_{B \in \mathcal B} \mu(B)$.
\end{theorem}

To show this, \cite{harris-llll} defined a type of witness tree, which differs slightly from the witness trees in the original analysis of Moser \& Tardos and from the stable-set sequences discussed in Appendix~\ref{app-kolm}. Let us summarize very briefly. Suppose we have run the sequential MT algorithm up to some time $T$, resampling bad-events $B_1, \dots, B_T$, and that some event $A$ is currently true. To generate the witness tree $\hat \tau_{A,T}$, we start with a root node labeled $A$. For each $\ell = T, T-1, \dots, 1$, we try to add a node to the tree with label $B_{\ell}$, placing it as a child of some node labeled $C$ with $C \sim B_{\ell}$. If there are multiple eligible positions we place the node at greatest depth (breaking ties arbitrarily).

However, one additional condition is enforced: for any node $v \in \hat \tau_{A,T}$ with label $C$, the children of $v$ must have distinct labels $C_1, \dots, C_s$ such that $\{C_1, \dots, C_s \}$ is orderable to $C$. A node $v$ is not eligible to have a child node labeled $B$, if adding such node would violate this condition. 

We say that a labeled tree $\tau$ \emph{appears} if $\hat \tau_{A,t} = \tau$ for any event $A$ and time $t$. We define $\mathfrak S'$ to be the set of all possible labeled trees that could appear.   The key lemma of \cite{harris-llll} is the following:
\begin{lemma}
\label{wt-lem-alt}
Any labeled tree $\tau$ appears with probability at most $w(\tau)$. Furthermore, we have $\sum \limits_{\tau \in \mathfrak S'} (1+ \epsilon)^{|\tau|} w(\tau) \leq W$ where we define $W = \sum_{B \in \mathcal B} \mu(B)$.  
\end{lemma}

Algorithm~\ref{par-alg1} can be viewed as a simulation of the sequential MT algorithm, so this same lemma applies to it.  By using arguments of \cite{harris-llll} for a similar parallel resampling algorithm, we can see that if a bad-event $B$ is true after $T$ total resamplings in the middle of round of $t$ of Algorithm~\ref{seq-alg1}, then witness tree $\hat \tau_{B, T}$ has depth $t$ and is rooted at $B$.  This allows us to show a result analogous to Lemma~\ref{length-bound} in terms of the orderability criterion, and thereby to show Theorem~\ref{alt-lll-thm}. Since the proof is nearly identical to Lemma~\ref{length-bound} and Theorem~\ref{par-bound2}, we omit it here.

\section{Proof of Theorem~\ref{atomic-thm}}
\label{atomic-thm-proof}
We suppose here we have a resampling-space $\mathcal A, R, U, \Omega, \sim$ satisfying conditions (C1), (C2), (C4). At later stages in the proof we may also assume it satisfies conditions (C3') and (C5).  

It will be convenient to work with ordered sequences from $\mathcal A$. We say that $H = (A_1, \dots, A_k)$ is a \emph{stable list} if $A_i \not \sim A_j$ for $i \neq j$.  For a permutation $\pi \in S_k$, we define $\pi H = (A_{\pi 1}, \dots, A_{\pi k})$. Likewise, we define $R_H$ to be the set of products $h_k \cdots h_1$ wherein $h_i \in R_{A_i; A_{i+1}, \dots, A_k}$.  Whenever we discuss resampling an event $C = \langle E \rangle$ and we write $E = \{A_1, \dots, A_k \}$, then we tacitly assume that we have chosen to order the elements of $E$ as $A_1, \dots, A_k$, so that $R_C = R_H$ for the stable list $H = (C_1, \dots, C_k)$. 
 
\begin{proposition}
  \label{c1prop}
 $\oa$ satisfies (C1).
\end{proposition}
\begin{proof}
Consider $C = \langle A_1, \dots, A_k \rangle$. Let $r_1, \dots, r_k$ be independent variables, wherein
$r_i$ is drawn from $\Gamma_{A_i; A_{i+1}, \dots, A_k}$. We need to show that when $u \approx \Omega | C$,
then $r_k \dots r_1 u \approx \Omega$.

For each $i = 0, \dots, k$ let us define $Q_i = A_{i+1} \cap \dots \cap A_k$ and $u_i =r_i \dots r_1 u$. Since each $r_i$ is chosen from $R_{A_i; Q_{i+1}}$, we see that $u_i \in Q_i$ with probability one for all $i$. We will show that that $u_i \approx \Omega | Q_i$ by induction on $i$. The base case $i = 0$ is given to us by hypothesis (since $C = A_1 \cap \dots \cap A_k$), and the case $i = k$ is what we are trying to prove. 

Consider a state $\tilde u \approx \Omega | A_i$ and $\tilde r \approx \Gamma_{A_i}$. For any $v \in U$, property (C1) gives $\Pr( \tilde r \tilde u = v ) = \Omega[v]$. If $\tilde r \tilde u \in Q_{i+1}$, then we claim that $\tilde u \in Q_{i+1}$; for, if $\tilde u \notin A_j$ for some $j > i$, then by property (C2) $\tilde r \tilde u \notin A_j$ as well. Similarly, if $\tilde r \tilde u \in Q_{i+1}$, then $\tilde r \in R_{A_i; A_j}$; for if $\tilde r \notin R_{A_i; A_j}$, then by property (C4) we would have $\tilde r \tilde u \notin A_j$. Thus, for $v \in Q_{i+1}$, we have
\begin{align*}
  \Omega[v] &= \Pr(\tilde r \tilde u = v) = \Pr(\tilde r \tilde u = v \wedge \tilde u \in Q_{i+1} \wedge \tilde r \in R_{A_i; \{ A_{i+1}, \dots, A_k \} } ) \\
  &= \Pr(\tilde r \tilde u = v \mid \tilde u \in Q_{i+1} \wedge \tilde r \in R_{A_i; \{ A_{i+1}, \dots, A_k \} } )  \Pr(\tilde u \in Q_{i+1}) \Pr(\tilde r \in R_{A_i; \{ A_{i+1}, \dots, A_k \} } )
\end{align*}

By induction hypothesis, $u_i$ and $\tilde u \mid \tilde u \in Q_{i+1}$ both have the distribution $\Omega |  Q_i$.  Likewise, $r_i$ and $\tilde r \mid \tilde r \in R_{A_i; Q_{i+1}}$ both have the distribution $\Gamma_{A_i; Q_{i+1}}$. Furthermore, the variables $\tilde u, \tilde r$ are independent and the variables $u_i, r$ are independent. This implies that
$$
 \Pr(\tilde r \tilde u = v \mid \tilde u \in Q_{i+1} \wedge \tilde r \in R_{A_i; \{ A_{i+1}, \dots, A_k \} } )  = \Pr( r_i u_i = v)
$$

So $\Omega[v] = \Pr( r_i u_i = v)  \Pr(\tilde u \in Q_{i+1}) \Pr(\tilde r \in R_{A_i; \{ A_{i+1}, \dots, A_k \} } )$. This shows that $\Pr(r_i u_i = v)$ is proportional to $\Omega[v]$ for any $v \in Q_{i+1}$. Since $r_i u_i \in Q_{i+1}$ with probability one, this implies that $u_{i+1} = r_i u_i \approx \Omega | Q_{i+1}$. 
\end{proof}

\begin{proposition}
 $\oa$ satisfies (C2).
\end{proposition}
\begin{proof}
  Consider $C = \langle A_1, \dots, A_k \rangle$ and $C' = \langle E' \rangle$ with $C \not \sim C'$, and let $u \in C - C'$. Consider $r = r_k \dots r_1  \in R_C$. There must exist some $A' \in E'$ such that $u \notin A'$.  We can show that that $r_i \dots r_1 u \notin A'$ for all $i$, by an induction on $i$: the base case $i = 0$ holds since $u \notin A'$, and the induction step follows from property (C2) applied to event $A_i$ and $A'$.

  At $i = k$, this shows that $r u = r_k \cdots r_1 u \notin A' \supseteq C$.
\end{proof}

\begin{proposition}
  \label{equiv-gg1}
  Let $C = \langle A_1, \dots, A_k \rangle$ and $C' = \langle A_{k+1}, \dots, A_{\ell} \rangle$ be events in $\oa$ where $C \not \sim C'$. For any state $u \in C \cap C'$ and $r  \in R_C$, the following are equivalent:
  \begin{enumerate}
  \item $r u \in C'$
  \item There exist $r_1, \dots, r_k$ such that  $r = r_k \cdots r_1$ and $r_i \in R_{A_i; A_{i+1}, \dots, A_k, A_{k+1}, \dots, A_{\ell}}$ for all $i = 1, \dots, k$
    \end{enumerate}
\end{proposition}
\begin{proof}
For (2) $\Rightarrow$ (1), a simple induction on $i$ shows that $r_i \cdots r_1 u \in C'$ for $i = 0, \dots, k$. 

For (1) $\Rightarrow$ (2), the definition of $R_C$ shows  $r = r_k \dots r_1$ where each $r_i$ is in $R_{A_i; A_{i+1}, \dots, A_k}$. If $r_i \in R_{A_i; E'}$ for $i = 1, \dots, k$ we are done; otherwise, let $i$ be minimal such that $r_i \notin R_{A_i; A_j}$ for some $j > k'$. So $u' = r_i \dots r_1 u \notin A_j$. Since $C \not \sim C'$, by repeated applications of (C2), we see also that $r_k \dots r_1 u = r_k \dots r_{j+1} u'$ is also not in $A_j$ and hence not in $C'$.
\end{proof}

\begin{corollary}  
  $\oa$ satisfies (C4).
\end{corollary}
\begin{proof}
  For events $C, C'$ with $C \not \sim C'$, Proposition~\ref{equiv-gg1} gives an explicit condition on $r \in R_C$  to ensure that $r u \in C'$ for $u \in C \cap C'$. This condition depends solely on $r$, and not $u$ itself.
\end{proof}

\begin{proposition}
\label{oac6}
If $\mathcal A$ satisfies (C5), then  $\oa$ satisfies (C5).
\end{proposition}
\begin{proof}
Consider $C = \langle E \rangle$ for $E = \{A_1, \dots, A_k \}$.  For each $i = 1, \dots, k$ we define $G_i = R_{A_i; A_{i+1}, \dots, A_k}$. For $i = k+1, \dots, 1$, we claim that there exists exactly one state $w_i \in A_{i} \cap \dots \cap A_k$ such that $u \in G_k \dots G_{i} w_i$. The base case $i =k+1$ holds vacuously with $w_i = u$, and the case $i = 1$ is what we are trying to show.

  For the induction step, we first show existence.  By (C5), there exists $w_{i} \in A_{i}$ such that $w_{i+1} \in R_{A_i; A_{i+1}} w_i$. So  $w_{i+1} = h w_i$ for some $h \in R_{A_i; A_{i+1}}$. By induction hypothesis, we have $w_{i+1} \in A_j$ for $j > i+1$. Since $A_i \not \sim A_j$, it must be the case that $w_i \in A_j$ and $r \in R_{A_i; A_j}$ for each such $j$. Thus, $h \in R_{A_i; A_{i+1}, \dots, A_k} = G_i$ and $w_i \in A_i \cap \dots \cap A_k$.

  Next, we show uniqueness. Suppose that $w_{i+1} \in G_i w'$ for some $w' \in A_i \cap \dots \cap A_k$. Since $w' \in A_i$ and $G_i \subseteq R_{A_i; A_{i+1}}$, by (C5) this implies that $w' = w_i$.
\end{proof}

 \begin{proposition}
   \label{oac5}
Suppose that $\mathcal A$ satisfies (C3'). Then for a stable list $H = (A_1, \dots, A_k)$, any $u \in U$, and any $\pi \in S_k$, we have $R_H u = R_{\pi H} u$.
\end{proposition}
\begin{proof}
Since we can generate any permutation $\pi$ by swapping adjacent elements, it suffices to show this holds when $\pi = (j  \  \ \ j+1)$ for some $j < k$. 

Let $r = h_k \cdots h_1 \in R_H$ wherein each $h_i \in R_{A_i; A_{i+1}, \dots, A_k}$. Define $u' = h_{j-1} \dots h_1 u$. Note that $u' \in A_j \cap A_{j+1}$. By (C3') applied to events $A_j, A_{j+1}$, there exist $h_j' \in R_{A_j}, h_{j+1}' \in R_{A_{j+1}; A_j}$ with $h'_j h'_{j+1} u' = h_{j+1} h_j u'$. Since $h_{j+1} h_j u' \in A_{j+2} \cap \dots \cap A_k$, it must be the case that $h'_j \in R_{A_{j+1}; A_{j+2}, \dots, A_k}$ and $h'_{j+1} \in R_{A_{j+1}; A_j, A_{j+2}, \dots, A_k}$.

Now set $r' = h_k h_{k-1} \dots h_{j+2} h_j' h_{j+1}' h_{j-1} \dots h_1$. We thus have shown that $r' \in R_{\pi H}$. Furthermore, we have $r u = h_k \dots h_1 u = h_k \dots h_{j+2} h_j' h'_{j+1} h_{j-1} \dots h_1 u = r' u$. \qedhere
\end{proof}

\begin{proposition}
  \label{hbbn1}
  If $\mathcal A$ satisfies (C3'), then $\oa$ satisfies (C3')
\end{proposition}
\begin{proof}
  Consider events $C_1 = \langle A_1, \dots, A_k  \rangle$ and $C_2 = \langle  A_{k+1}, \dots, A_{\ell} \rangle$ and any $u \in C_1 \cap C_2$. By symmetry, it suffices to show that for any $r_1 \in R_{C_1; C_2}, r_2 \in R_{C_2}$ there are $r_1' \in R_{C_1}, r_2' \in R_{C_2}$ with $r_2 r_1 u = r_1' r_2' u$.
  
Define $H = (A_1, \dots, A_{\ell})$. By definition of $R_{C_2}$, we have $r_2 = h_{\ell} \cdots h_{k+1}$ where $h_i \in R_{A_i; A_{i+1}, \dots, A_{\ell}}$ for $i = k+1, \dots, \ell$. By Proposition~\ref{equiv-gg1}, we have $r_1 = h_k \cdots h_1$ where $h_i \in R_{A_i; A_{i+1}, \dots, A_k, A_{k+1}, \dots, A_{\ell}}$ for $i = 1, \dots, k$. Thus, we see that $r_2 r_1 \in R_H$.

  Now define $H' = (A_{k+1}, \dots, A_{\ell}, A_1, \dots, A_k)$ and note that $H'$ is a rearrangement of the list $H$. By Proposition~\ref{oac5}, this implies that there exists $r' \in R_{H'}$ such that $r' u = r_2 r_1 u$. We can write $r' = h_k' \dots h_1' h_{\ell}' \dots h_{k+1}'$, wherein $h'_i \in R_{A_i; A_1, \dots, A_k, A_{i+1}, \dots, A_{\ell} }$ for $i = k+1, \dots, \ell$, and $h'_i \in R_{A_i; A_{i+1}, \dots, A_k}$ for $i = 1, \dots, k$. If we set $r'_1 = h'_k \dots h'_1$ and $r_2' = h'_{\ell} \dots h'_{k+1}$, then $r_1' \in R_{C_1}$ and by Proposition~\ref{equiv-gg1} we have $r_2' \in R_{C_2; C_1}$.  We then have $r_2 r_1 u = r_1' r_2' u$ as desired.  
\end{proof}

\section{Proof of Theorem~\ref{greedy-th0}}
\label{lfmis-app}
Consider a directed graph $G = (V,E)$, with a permutation $\pi: [n] \rightarrow V$ chosen uniformly at random. Let $G^{\pi}$ denote the directed acyclic graph on vertex set $V$ and edge-set $\{(u,v) \mid (u,v) \in E, \pi^{-1}(u) < \pi^{-1}(v) \}$.  Let $I^{\pi}$ denote the LFMIS of $G$ with respect to $\pi$.  For any integer $j \in [n]$, define the partial LFMIS $I^{\pi}_j = I^{\pi} \cap \{ \pi^{-1}(1), \dots, \pi^{-1}(j) \}$. For integers $0 \leq i \leq j \leq n$, define the residual vertex set $V^{\pi}_{(i,j]} = \{ \pi^{-1}(i+1), \dots,  \pi^{-1}(j) \} - I^{\pi}_{i} - \bigcup_{v \in I^{\pi}_i} N^{\text{out}}(v)$ and define $G^{\pi}_{(i,j]}$ to be the induced subgraph $G^{\pi}[V^{\pi}_{(i,j]}]$.

For the purpose of analysis, it will be useful to consider a slowed-down variant of Algorithm~\ref{algo:greedylfmis} called SLOW-GREEDY, as discussed in \cite{greedy2}.  Given integers $n_0, n_1, \dots, n_k$, it is defined as follows:
\begin{algorithm}[H]
\centering
\begin{algorithmic}[1]
\State Initialize $I = \emptyset$ and $A = V$
\For{$i = 1, \dots, k$}
\While{$ A \cap V_{(n_{i-1}, n_i]} \neq \emptyset$}
  \State Let $J$ be the set of nodes $v \in A \cap V_{(n_{i-1}, n_i]}$ such that $A \cap N^{\text{in}}(v) \cap P^{\pi}(v) = \emptyset$.
  \State Update $I \leftarrow I \cup J$ and $A \leftarrow A - \bigcup_{v \in J} N^{\text{out}}( v )$
\EndWhile
\EndFor
\end{algorithmic}
\caption{The SLOW-GREEDY algorithm}
\label{algo:slow-greedy}
\end{algorithm}

We refer to the $i^{\text{th}}$ iteration of the loop in line (2) as \emph{epoch $i$}. We make the following observations for Algorithm~\ref{algo:slow-greedy}; since the proofs are completely analogous to the undirected case, we refer to the reader to \cite{greedy2} for full proof details.
\begin{proposition}[\cite{greedy2}]
  \label{greedyprop}
For any integers $n_0, n_1, \dots, n_k$ with $0 = n_0 \leq n_1 \leq n_2 \leq \dots \leq n_k = n$, we have the following:
  \begin{enumerate}
  \item SLOW-GREEDY computes the LFMIS of $G$ with respect to $\pi$.
  \item The number of rounds in Algorithm~\ref{algo:greedylfmis} on $G$ and $\pi$ is at most the total number of rounds in SLOW-GREEDY.
    \item If all directed paths in $G^{\pi}_{(n_{i-1}, n_i]}$ have length at most $\ell$, then epoch $i$ of SLOW-GREEDY terminates in at most $\ell$ rounds.
  \end{enumerate}
\end{proposition}

Algorithm~\ref{algo:greedylfmis} can be viewed as a special case of SLOW-GREEDY with $n_0 = 0, n_1 = n, k = 1$; in particular, this shows that Algorithm~\ref{algo:greedylfmis} correctly computes the LFMIS of $G$ with respect to $\pi$.

We now analyze the path lengths in the subgraphs $G^{\pi}_{(i,j]}$. For $i = 0, \dots, n$, let us define
  $$
  D_i = \max_{v \in V^{\pi}_{(i,n]}}  |N^{\text{in}}(v) \cap V^{\pi}_{(i,n]}|
      $$

  \begin{proposition}
    \label{yw1}
    With probability at least $1 - n^{-100}$, we have $D_i \leq  \frac{200 n \log n}{i}$ for any $i = 1, \dots, n$.
\end{proposition}
\begin{proof}
 Let us fix some vertex $v$, and we want to show that either $v \notin V_{(i,n]}$ or $|N^{\text{in}}(v) \cap V_{(i,n]}| \leq d$ for $d = \frac{200 n \log n}{i}$. For each $k = 1, \dots, n$ define $\mathcal E_k$ to be the event that $v$ is alive and has at least $d$ alive in-neighbors after step $k$ of Algorithm~\ref{seq-lfmis1}.

We compute the probability of $\mathcal E_k$ conditional on $\mathcal E_1, \dots \mathcal E_{k-1}$. As $\mathcal E_1, \dots, \mathcal E_{k-1}$ are determined by $\pi(1), \dots, \pi(k-1)$, it suffices to compute the probability of $\mathcal E_i$ conditional on $\pi(1), \dots, \pi(k-1)$. This allows us to determine the set $A' = A \cap N^{\text{in}}(v)$ of alive in-neighbors of $v$ after step $k-1$. If $|A'| < d $,  then $\mathcal E_k$ is false. Otherwise, we have $\pi(k) \in A'$ with probability at least $\frac{d}{n-k+1}$, in which case $v$ is removed from $A$ after iteration $k$ and $\mathcal E_k$ is false. Thus,  $\Pr(\mathcal E_k \mid \mathcal E_1, \dots, \mathcal E_{k-1}) \leq 1 - \frac{d}{n-k+1}$. This implies that
$$
\Pr(\mathcal E_i) \leq (1 - \tfrac{d}{n}) (1 - \tfrac{d}{n-1}) \dots (1 - \tfrac{d}{n-k+1}) \leq \bigl( 1 - \tfrac{d}{n} \bigr)^{i} \leq e^{-d i/n} = e^{-\frac{200 i \log n}{i} } = n^{-200}
$$

By definition $V_{(i,n]}^{\pi}$ contains only vertices which are alive after iteration $i$. Thus, if  $\mathcal E_{i}$ is false, the desired property holds for $v$ and $i$. To finish, taking a union bound over all $n^2$ values of $v, i$.
\end{proof}

\begin{proposition}
  \label{yw2}
  Suppose that we condition on $\pi(1), \dots, \pi(i)$, and let $s = D_i j/n$, and let $L$ denote the length of the longest path in $G^{\pi}_{(i,j]}$. Then, with probability at least $1 - n^{-5}$, it holds that
  $$
  L \leq \begin{cases}
  O(s) & \text{if $s \leq \log n$} \\
   O( \frac{\log n}{\log \frac{2 \log n}{s}}) & \text{if $s > \log n$} 
   \end{cases}
   $$
    \end{proposition}
\begin{proof}
Consider the induced graph $H = G[V^{\pi}_{(i,n]}]$, which depends only on the values $\pi(1), \dots, \pi(i)$. Let $d = D_i$ be the maximum in-degree of $H$. We can enumerate the length $k$ paths of $H$ by choosing the final vertex in the path ($n$ choices), and each of the $k$ previous vertices in the path ($d$ choices each), so the number of length $k$-paths in $H$ is at most $n \times d^{k-1}$.

 A necessary condition for a path $v_1, \dots, v_k$ to survive to $G^{\pi}_{(i,j]}$ is that $\pi(v_1) < \pi(v_2) < \dots < \pi(v_k) \leq j$. Having conditioned on $\pi(1), \dots, \pi(i)$, this event has probability 
\begin{align*}
\frac{1}{k!}  \times \frac{j-i}{n-i} \times \frac{j-i-1}{n-i-1} \times \dots \times \frac{j-i-k+1}{n-i-k+1} \leq (j/n)^k/k!
\end{align*}

Taking a union-bound over all such paths, we have
\begin{align*}
\Pr(L \geq k) &\leq n d^{k-1} (j/n)^k/k! \leq n s^k / k! \leq (e s/k)^k
\end{align*}

If $s > \log n$, then note that for $k = 2 e s$ this is at most $2^{-k} \leq n^{-10}$. If $s \leq \log n$, then set $x = \frac{2 \log n}{s} \geq 2$ and $k = \frac{10 \log n}{\log x} \leq O(s)$; we then have $(e s/k)^k  = \exp\Bigl( \frac{-10 \log n}{\log  x} \times \log \bigl( \frac{10 \log n}{e s \log x } \bigr) \Bigr) = \exp\Bigl( \frac{-10 \log n}{\log  x} \times \log \bigl( \frac{5 x}{e \log x }\bigr) \Bigr)$.     As $x \geq 1$, standard analysis shows that $\log \bigl( \frac{5 x}{e \log x} \bigr) \geq 0.5 \log x$ for $x \geq 1$. Thus, this is at most $\exp \bigl( \frac{-10 \log n}{\log x} \times 0.5 \log x \bigr) = e^{-5 \log n} = n^{-5}$.
\end{proof}

We are now ready to bound the runtime. We show a slightly tighter bound in terms of the maximum in-degree of graph $G$.
\begin{theorem}
\label{greedy-th2}
Let $d = \max_{v \in G} |N^{\text{in}}(v)|$. When $\pi$ is chosen uniformly at random, then:
\begin{enumerate}
  \item For $d \leq \log n$, Algorithm~\ref{algo:greedylfmis} takes $O \Bigl( \frac{\log n}{\log \tfrac{2 \log n}{d} } \Bigr)$ rounds whp.
  \item For $d > \log n$, Algorithm~\ref{algo:greedylfmis} takes $O( \log n \log \tfrac{2 d}{\log n})$ rounds whp.
\end{enumerate}

In particular,  Algorithm~\ref{algo:greedylfmis} takes $O( \log d \log n ) \leq O(\log^2 n)$ rounds whp.
\end{theorem}

\begin{proof}
\myparagraph{1.} By Proposition~\ref{yw2} applied at $i = 0, j = n$, whp the graph $G^{\pi}_{(0,n]}$ has maximum path length $O(\frac{\log n}{\log \frac{2 \log n}{s}})$ where $s = D_i \leq d \leq \log n$.  By Proposition~\ref{greedyprop}, this implies that Algorithm~\ref{algo:greedylfmis} terminates in $O( \frac{\log n}{\log \tfrac{2 \log n}{d} })$ rounds whp.
    
    \myparagraph{2.} We will use Proposition~\ref{greedyprop} with parameters $k = \lceil \log_2  \frac{4d}{\log n}  \rceil$ and  $n_j = \min(n, \frac{2^j n \log n}{d})$ for $j = 1, \dots, k$ and $n_0 = 0$. Note that $n_k = n$ as required, since $\frac{2^k \log n}{d} \geq \frac{4d}{\log n} \times \frac{\log n}{d} \geq 4$. 

    Define $s_i = D_{n_{i-1}} n_i / n$ for $i = 1, \dots, k$.  For $i = 1$, we have $s_i \leq d n_1 /n \leq d \times \frac{2 n \log n}{n d} \leq \log n$. For $i \geq 2$, Proposition~\ref{yw1} shows that $D_{n_{i-1}} \leq \frac{200 n d \log n}{2^{i-1} n \log n} = O(d/2^i)$ with probability at least $1 - n^{-100}$, in which case $s_i \leq  O(d/2^i) \times ( 2^i n \log n / d ) / n = O( \log n)$. When these events occur,  then by Proposition~\ref{yw2}, each graph $G^{\pi}_{(n_{i-1}, n_i]}$ for $i \geq 1$ has maximum path length $O(\log n)$ with probability at least $1 - n^{-5}$.

By Proposition~\ref{greedyprop}, these facts imply that, whp, each epoch of SLOW-GREEDY takes $O(\log n)$ rounds. Overall, the total number of rounds over all $k$ epochs is $O( k \log n) = O( \log n \log \tfrac{2 d}{\log n})$. \qedhere
\end{proof}

\section{Hamiltonian cycles of $K_n$}
\label{hamilton-sec}
In order to use algebraic tools, we encode a hamiltonian cycle $(x_1, \dots, x_n, x_1)$ of $K_n$  as the permutation  $\pi = (x_1 \ x_2 \ x_3 \dots \ x_n )$. In this way, the ground set $U$ can be viewed as the set of permutations $\pi$ consisting of precisely one cycle of length $n$. We define $R$ to be the group $S_n$ with the natural group action of left-multiplication on $U$; thus properties (D0), (D2), (D3) are trivial.

For any sequence of distinct values $x_1, \dots, x_k$, let us define the set of permutations
$$
T( x_1, \dots, x_k ) = \bigl \{ (x_k \ z_k) \cdots (x_1 \ z_1) \mid z_i \in [n] - \{x_i, \dots, x_k \} \bigr \}
$$

Note that each choice for the values for $z_1, \dots, z_k$ give rise to a distinct permutation. Thus, $| T( x_1, \dots, x_k  ) | = \frac{ (n-1)! }{ (n-k-1)! }$.

We are now ready to define the resampling-space itself. Let $Q$ be the set of paths $q = (x_1, \dots, x_k)$ where $x_1, \dots, x_k$ are distinct elements of $[n]$.  We define the \emph{support} of the path $q$ by $\sup(q) = \{x_1, \dots, x_k \}$.  For such path $q$, define an atomic event
$$
\langle q \rangle = \bigl \{ \pi \in U \mid \pi(x_1) = x_2, \dots, \pi(x_{k-1}) = x_k \bigr \}
$$

We define the dependency relation by setting $\langle q \rangle \sim \langle q' \rangle$ if $\sup(q) \cap \sup(q') \neq \emptyset$.

For a given set $X \subseteq [n]$, let us define $U_X$ to be the set of permutations in $S_n$ whose cycle structure consists of fixed points at each $x \in X$, along with a single cycle on $[n] - X$. Note that $U = U_{\emptyset}$. There is an important permutation which ``normalizes'' the path $q = (x_1, \dots, x_k)$, namely
$$
\lambda_q = (x_k \ x_{k-1} \dots x_1)
$$

For $q = (x_1, \dots, x_k)$, we define $\Gamma_{\langle q \rangle}$ to be to the uniform distribution on $T( x_1, \dots, x_{k-1} ) \lambda_q$. The following observations explain the role of $\lambda_q$:
\begin{observation}
\label{gb30}
For $\pi \in S_n$ and path $q = (x_1, \dots, x_k)$, we have $\pi \in \langle q \rangle$ iff $\lambda_{q} \pi \in U_{ \{x_1, \dots, x_{k-1} \}}$.
\end{observation}
\begin{proposition}
  \label{c0-prob-3aa}
  Let $A = \langle (x_1, \dots, x_{k}) \rangle$. For $\pi \in A$ and $\sigma \lambda_q \in R_A$, we have $\sigma \lambda_q \pi \in U$.
  \end{proposition}
\begin{proof}
  Let $\sigma = (x_{k-1} \ z_{k-1}) \cdots (x_1 \ z_1)$ where $z_i \in [n] - \{x_i, \dots, x_{k-1} \}$, and $\tau_i = (x_i \ z_i) \cdots (x_1 \ z_1) \lambda_q \pi$ for $i = 0, \dots, k-1$.  We  show by induction on $i$ that $\tau_i \in U_{ \{x_{i+1}, \dots, x_{k-1} \} }$. The base case at $i = 0$ is precisely Observation~\ref{gb30} since $\tau_0 = \lambda_q \pi$, and the case at $i = k-1$ is what we are trying to show since $\sigma = \tau_{k-1}$ and $U_{\emptyset} = U$.

  For the induction step, we have $\tau_i = (x_i \ z_i) \tau_{i-1}$. The point $x_i$ does not appear in the cycle of $\tau_{i-1}$ by induction hypothesis. However, since $z_i \in [n] - \{ x_i, \dots, x_{k-1} \}$, the point $z_i$ does so. Thus $\tau_i$ has $x_i$ inserted just before $z_i$ in its cycle, moving $x_i$ from a fixed point to part of its cycle.
\end{proof}

We now show that the necessary properties are satisfied.
\begin{proposition}
\label{c3-prop-5gg}
Properties (C5) and (C1) hold.
\end{proposition}
\begin{proof}
Consider $A = \langle q \rangle$ for a path $q = (x_1, \dots, x_{k})$ and let $\rho \in U$. We claim that there is precisely one choice for the ordered pair $(\sigma, \pi)$ with $\sigma \in T( x_1, \dots, x_{k-1} )$ and $\pi \in A$ such that $\rho = \sigma \lambda_q \pi$. 

Since $\pi$ is uniquely determined from $\rho, \sigma$, we will show that there is precisely one choice for $\sigma$ such that $ \sigma^{-1} \rho \in \lambda_q A$. By Observation~\ref{gb30}, this is equivalent to showing  $\sigma^{-1} \rho \in U_{ \{x_1, \dots, x_{k-1} \}}$.

Consider $\sigma = (x_{k-1} \ z_{k-1}) \dots (x_1 \ z_1)$ where $z_i \in [n] - \{x_i, \dots, x_k \}$.  We want to show that there is a unique choice for indices $z_1, \dots, z_{k-1}$ such that $\sigma^{-1} \rho = (x_1 \ z_1) \dots (x_{k-1} \ z_{k-1}) \rho$ is in $U_{ \{x_1, \dots, x_{k-1} \}}$.

It suffices to show that for any index $j = k-1, \dots, 1$ and $\tau \in U_{ \{x_{j+1}, \dots, x_{k-1} \}}$, there is a unique choice for $z_j$ such that $(x_j \ z_j) \tau \in U_{ \{x_j, \dots, x_{k-1} \} }$. Since $\tau \in U_{ \{x_{j+1}, \dots, x_{k-1} \}}$, the element $x_j$ appears in the full cycle, followed by some $y \notin \{x_{j+1} \dots, x_{k-1} \}$. Now note that $(x_j \ z_j) \tau$ has an additional fixed point at $x_j$ precisely if $z_j = y$. Thus there is precisely one choice of $z_j$ with $(x_j \ z_j) \in U_{ \{x_j, \dots, x_{k-1} \} }$.

This shows the claim and immediately gives (C5). For (C1), note  that for any $\rho \in U$, the probability of  $\rho = \sigma \lambda_q \pi$, where $\sigma$ is drawn uniformly  from $T(x_1, \dots, x_{k-1})$ and $\pi$ is drawn uniformly from $A$, is precisely $\frac{1}{|T(x_1, \dots, x_{k-1})|} \times \frac{1}{|A|} = \frac{(n-k-1)!}{(n-1)!} \times \frac{1}{(n-k-1)!} = \frac{1}{(n-1)!}$.
\end{proof}

\begin{proposition}
\label{c4-prop-3}
Property (C2) holds.
\end{proposition}
\begin{proof}
  Consider $A = \langle q \rangle$ for $q = (x_1, \dots, x_{k})$ and $A' = \langle q' \rangle$ for $q' = (y_1, \dots, y_j)$ with $A \not \sim A'$ and $\pi \in A - A'$. There must exist some index $\ell < i$ with $\pi(y_{\ell}) \neq y_{\ell+1}$.

  Let $\sigma \in T( x_1, \dots, x_{k-1} )$. We claim that $\sigma \lambda_q \pi y_{\ell} \neq y_{\ell+1}$ so that $\sigma \lambda_q \pi \notin A'$. 

To show this, define $\tau_i = (x_i \ z_i) \cdots (x_1 \ z_1) \lambda_q \pi$ for $i =0, \dots, k-1$, wherein $z_j \in [n] - \{x_j, \dots, x_{k-1} \}$. Suppose that $i$ is minimal such that $\tau_i y_{\ell} = y_{\ell+1}$. It cannot be $i = 0$, as $\lambda_q y_{\ell+1} = y_{\ell+1}$ (since $y_{\ell+1} \notin \sup(q)$). 

For this value $i > 0$, it must be either that  (a) $x_i = \tau_{i-1} y_{\ell}, z_i = y_{\ell+1}$ or (b) $z_i = \sigma_{i-1} y_{\ell}, x_i = y_{\ell+1}$. The former cannot occur as $\tau_{i-1} x_i = x_i$ and the latter cannot occur as $x_i \neq y_{\ell+1}$.
\end{proof}

\begin{proposition}
  \label{c6-prop-3}
 Let $q = (x_1, \dots, x_{k})$ and $b \in [n] - \{x_1, \dots, x_{k} \}$. Let $\sigma = (x_{k-1} \ z_{k-1})  \cdots (x_1 \ z_1)$ where $z_i \in [n] - \{x_i, \dots, x_{k-1} \}$. Then $\sigma b = b$ iff $z_1, \dots, z_{k-1}$ are all distinct from $b$.
\end{proposition}
\begin{proof}
  The reverse direction is immediate.   For the forward direction, define $\sigma_j = (x_j \ z_j) \cdots (x_1 \ z_1)$ for $j = 0, \dots, k-1$ and let $i \leq k-1$ be minimal such that $z_i = b$.   We show by induction that for $j \geq i$ we have $\sigma_j b \in \{x_1, \dots, x_{k-1} \}$. For the base case, we have $\sigma_i b = (x_i \ b) (x_{i-1} \ z_{i-1}) \cdots (x_1 \ z_1) b = x_i$.  For the induction  step, suppose that $\sigma_{j-1} b = x_r$. If $z_i \neq x_r$ we have $\sigma_j b = \sigma_{j-1} b = x_r$ as desired. If $z_j = x_r$, then $\sigma_j b = (x_j \ x_r) \sigma_{j-1} x_r = x_j$, again as desired.

Thus, if some of the $z_i$ are equal to $b$ then $\sigma b \in \{x_1, \dots, x_{k-1} \}$, and in particular $\sigma b \neq b$.
\end{proof}

\begin{proposition}
\label{c6x-prop-3}
Property (C4) holds. Furthermore, for $A = \langle q \rangle, A' = \langle q' \rangle$ with $A \not \sim A', q = (x_1, \dots, x_k), q' = (y_1, \dots, y_j)$, we
have
  $$
  R_{A;A'} =\bigl \{ (x_{k-1} \ z_{k-1}) \cdots (x_1 \ z_1) \lambda_q \mid z_i \in [n] - \{y_2, \dots, y_j, x_i, \dots, x_{k-1} \} \bigr \} 
  $$  
\end{proposition}
\begin{proof}
 Let $\ell < j$. Consider $\sigma = (x_{k-1} \ z_{k-1}) \cdots (x_1 \ z_1) \lambda_q \in R_A$. For $\pi \in A'$, we have $\sigma \lambda_q \pi y_{\ell} = \sigma \lambda_q y_{\ell+1} = \sigma y_{\ell+1}$; by Proposition~\ref{c6-prop-3} this is equal to $y_{\ell+1}$ iff $z_1, \dots, z_{k-1}$ are distinct from $y_{\ell+1}$. Thus, $\sigma \lambda_q \pi \in A'$ iff $z_1, \dots, z_{k-1}$ are distinct from $y_2, \dots, y_j$.  To show (C4), note that this criterion does not depend on $\pi$, so it either holds for all $\pi \in A \cap A'$ or none of them.
\end{proof}

Given any event $A = \langle (x_1, \dots, x_k) \rangle$ and stable set $E \not \sim A$, this result allows us to efficiently draw from $R_{A;E}$, by selecting indices $z_2, \dots, z_k$ wherein each $z_i$ is distinct from the tail $y_2, \dots, y_j$ for each $A' = \langle (y_1, \dots, y_j)$ in $E$. In particular, this shows (D1'). 

We will now show commutativity. This follows from the observation that $T( x_1, \dots, x_k)$ depends only on the \emph{unordered} set $\{x_1, \dots, x_k \}$:
\begin{proposition}
\label{c0-prop-3b}
For any distinct values $x_1, \dots, x_k$ and any permutation $\pi \in S_k$, we have
$$
T( x_1, \dots, x_k) = T( x_{\pi 1}, \dots, x_{\pi k})
$$
\end{proposition}
\begin{proof}
 It suffices to consider $\pi = (j \ j+1)$ for $j < k$.  Consider $\sigma = (x_k \ z_k ) \cdots (x_1 \ z_1)$ where $z_i \in [n] - \{x_i, \dots, x_k \}$.  We will show that there exist $w_j, w_{j+1}$ such that $(x_j \ w_j) (x_{j+1} \ w_{j+1}) = (x_{j+1} \ z_{j+1}) (x_j \  z_j)$ with $w_j \notin \{x_j, x_{j+2},  \dots, x_k \}, w_{j+1} \notin \{x_j, x_{j+1}, x_{j+2}, \dots, x_k \}$. In this case, replacing the terms $(x_{j+1} \ z_{j+1}) ( x_j \  z_j )$ with $(x_j \ w_j)(x_{j+1} \ w_{j+1})$ allows us to swap $x_j, x_{j+1}$, showing that $\sigma \in T( x_1, x_2, \dots, x_{j-1}, x_{j+1}, x_j, x_{j+2}, \dots, x_k )$.  There are a few cases.
\begin{enumerate}
\item If all four values $z_j, z_{j+1}, x_j, x_{j+1}$ are distinct, then $(x_{j+1} \ z_{j+1}) (x_j \ z_j) = (x_{j+1} \ z_{j+1}) (x_j \ z_j)$ and so $w_j = z_j, w_{j+1} = z_{j+1}$ works.
\item If $z_j = z_{j+1} = z$, then $(x_{j+1} \ z_{j+1}) (x_j \ z_j) = (x_j \ x_{j+1} \ z) = (x_j \ x_{j+1}) (x_{j+1} \ z)$. Thus taking $w_j = x_{j+1}$ and $w_{j+1} = z$ works.
\item If $z_{j+1} = x_j$, then $(x_{j+1} \ z_{j+1}) (x_j \ z_j) = (x_j \ z_j \ x_{j+1}) = (x_j  \ z_j)(x_{j+1} \ z_j)$. Thus taking $w_j = z_j, w_{j+1} = z_j$ works. \qedhere
\end{enumerate}
\end{proof}

\begin{proposition}
Property (C3') holds.
\end{proposition}
\begin{proof}  
  Let $A_1 = \langle q_1 \rangle, A_2 = \langle q_2 \rangle$ where $q_1 = (x_1, \dots, x_{k}), q_2 = (b_1, \dots, b_{\ell})$ with $A_1 \not \sim A_2$. We will show that
\begin{equation}
\label{th-eqn}
R_{A_2} R_{A_1; A_2}  =  T(H) \lambda_q \lambda_{q'}
\end{equation}
where we define $H = (x_1, \dots, x_{k-1}, b_1, \dots, b_{\ell-1})$. Note that $\lambda_q$ and $\lambda_{q'}$ commute since $A_1 \not \sim A_2$, and by Proposition~\ref{c0-prop-3b} the set $T(H)$ does not depend upon the ordering of the list $H$, and so by symmetry this will then show that $R_{A_2} R_{A_1; A_2}  =  T(H) \lambda_q \lambda_{q'} = R_{A_1} R_{A_2; A_1}$ as desired.

Since $A_1 \not \sim A_2$, the values $b_1, \dots, b_{\ell}$ are distinct from $x_1, \dots, x_k$. We have $|R_{A_2}| = \frac{(n-1)!}{(n-\ell)!}$ and  $|T(H)| = \frac{(n-1)!}{(n-1-(\ell + k - 2))!}$. Using the explicit description of $R_{A_1;A_2}$ from Proposition~\ref{c6x-prop-3}, we calculate $|R_{A_1;A_2}| = \frac{(n-1-(\ell-1))!}{(n-1-(\ell-1)-(k-1))!}$. Thus $| R_{A_2} | \times | R_{A_1;A_2} | =  |T(H)|$. We will show that $T(H) \lambda_q \lambda_{q'}  \subseteq R_{A'} R_{A;A'}$;   a counting argument then shows Eq.~(\ref{th-eqn}).

Consider $\tau \in T(H)$ of the form 
$$
\tau = (b_{\ell-1} \ c_{\ell-1}) \cdots (b_1 \ c_1) (x_{k-1} \ z_{k-1}) \cdots (x_1 \ z_1),
$$ where $z_i \notin \{x_i, \dots, x_{k-1}, b_1, \dots, b_{\ell-1} \}$ and $c_i \notin \{b_i, \dots, b_{\ell-1} \}$.

If $z_i \neq b_1$, then $\lambda_{q'} (x_i \ z_i) = (x_i \ z_i) \lambda_{q'}$. Otherwise, for $z_i = b_1$, we have $\lambda_{q'} (x_i \ z_i) = \lambda_{q'} (x_i \ b_1) = (x_i \ b_{\ell} \dots b_1) = ( x_i  \ b_{\ell} ) \lambda_{q'}$. This shows that $\lambda_{q'} (x_{k-1} \ z_{k-1}') \cdots (x_1 \ z_1') = (x_{k-1} \ z_{k-1}) \cdots (x_1 \ z_1) \lambda_{q'}$, where $z'_i$ is defined as $$
z'_i = \begin{cases}
b_1 & \text{if $z_i = b_{\ell}$} \\
z_i & \text{otherwise}
\end{cases}
$$

So we have shown that $\tau \lambda_q \lambda_{q'} = (b_{\ell-1} \ c_{\ell-1}) \cdots (b_1 \ c_1) \lambda_{q'} (x_{k-1} \ z_{k-1}') \cdots (x_1 \ z_1') \lambda_q$. Since $z_i \notin \{x_i, \dots, x_{k-1}, b_1, \dots, b_{\ell-1} \}$, likewise $z'_i \notin \{x_i, \dots, x_{k-1}, b_2, \dots, b_{\ell} \}$.  So, by Proposition~\ref{c6x-prop-3} we have $(x_{k-1} \ z'_{k-1}) \cdots (x_1 \ z'_1) \lambda_{q'} \in R_{A;A'}$. Clearly, $(b_{\ell-1} \ c_{\ell-1}) \cdots (b_1 \ c_1) \lambda_q  \in R_A$. So we have shown that $\tau \lambda_q \lambda_{q'}$ can indeed be written as an element of $R_{A_2} R_{A_1; A_2}$.
\end{proof}

\section{Perfect matchings of $K_n^{(s)}$}
\label{match-sec2}
Let us fix $s \geq 2$ throughout this section and $n$ a multiple of $s$ and we define $U = \mathcal M$ to be the set of perfect matchings of $K_n^{(s)}$. Note that the case $s = 2$ is the space of perfect matchings of $K_n$, which has been studied more extensively, with a commutative resampling oracle given by Kolmogorov \cite{kolmogorov}. In \cite{mohr}, Lu, Sz\'{e}k\'{e}ly \& Mohr showed (non-algorithmically) that the LLLL held for all $s \geq 2$. 

We will construct an oblivious resampling-space for the uniform distribution on $\mathcal M$. This gives efficient sequential algorithms. We also show that when $s = 2$, the space is commutative and is compatible with our parallel algorithm. 

The probability space $\Omega$ is the uniform distribution on $\mathcal M$. For every size-$s$ subset $e$ of $[n]$, we define the atomic event
$$
\langle e \rangle = \bigl \{ M \in \mathcal M \mid e \in M  \bigr \}
$$

The dependency relation $\sim$ is defined by setting $\langle e \rangle \sim \langle e' \rangle$ iff $e \neq e'$ and $e \cap e' \neq \emptyset$.

The monoid $R$ is the symmetric group $S_n$, with the natural group action on $U$ defined by 
$$
\sigma M  = \bigl \{ \{\sigma x_1, \dots, \sigma x_s \} \mid \{ x_1, \dots, x_s \} \in M \bigr \}
$$
It is clear that properties (D0), (D2), (D3) hold.

Whenever we enumerate an edge  $e = \{x_1, x_2, \dots, x_s \}$, we always assume implicitly it is sorted so that $x_1 < x_2 < \dots < x_s$. With this notation in mind, for an event $A = \langle \{x_1, \dots, x_s \} \rangle$ we define the set of permutations $$
R_A = \bigl \{ (x_2 \ z_2) \dots (x_s \ z_s) \mid z_i \in [n] - \{x_1, \dots, x_{i-1} \} \bigr \}
$$
and we define $\Gamma_A$ to be the uniform distribution on $R_A$. Note that each choice of $z_2, \dots, z_s$ gives rise to a distinct permutation, so that $\Gamma_A$ also corresponds to the distribution obtained by choosing each index $z_i$ independently and uniformly from the range the $[n] - \{x_1, \dots, x_{i-1} \}$.

\begin{proposition}
\label{c1-str}
For any event $A = \langle e \rangle$ and any $N \in \mathcal M$,  there are precisely $(s-1)!$ ordered pairs $(\sigma, M) \in R_A \times A$ such that $\sigma M = N$.  In particular, for $s = 2$, property (C5) holds.
\end{proposition}
\begin{proof}
Let $e = \{x_1, \dots, x_s \}$.  Since $M$ is uniquely determined from $\sigma, N$ it suffices to show there are precisely $(s-1)!$ choices for $\sigma$ such that $\sigma^{-1} N \in A$.

Consider $\sigma = (x_2 \ z_2) \cdots (x_s \ z_s)$ where $z_i \in [n] - \{x_1, \dots,x_{i-1} \}$. For each $j = 1, \dots, s$ let us define $A_j$ to be the set of matchings $M$ such that $\{x_1, \dots, x_i \} \subseteq e$ for some $e \in M$.  We claim that, given any matching $M \in A_j$, there are precisely $s-j$ choices for $z_{j+1} \in [n] - \{x_1, \dots, x_j \}$ such that $(x_{j+1} \ z_{j+1}) M \in A_{j+1}$. As $N \in A_1 = \mathcal M$ and $A_s = s$, this will establish that there are precisely $(s-1) \cdots 1 = (s-1)!$ choices for $z_2, \dots, z_s$ such that $(x_s \ z_s) \cdots (x_2 \ z_2) N = \sigma^{-1} N$ is in $A$.

Now suppose we have chosen values $z_2, \dots, z_j$, and so $N' = (x_j \ z_j) \dots (x_2 \ z_2) N$ has been determined.  By hypothesis, $N' \in A_j$ and so $N'$ contains an edge $e = \{x_1, \dots, x_j, y_1, \dots, y_{s-j} \}$. We have $(x_{j+1} \ z_{j+1}) N' \in A_{j+1}$ iff $x_{j+1}$ is swapped into edge $e$, which occurs precisely when $z_{j+1} \in \{y_1, \dots, y_{s-j} \}$. Thus, there are $s-j$ choices for $z_{j+1}$ as we have claimed. 
\end{proof}

\begin{proposition}
Property (C1) holds.
\end{proposition}
\begin{proof}
  Consider event $A = \langle e \rangle$.  By Proposition~\ref{c1-str}, there are precisely $(s-1)!$ pairs  $\sigma \in R_A,  M \in A$ which lead to a given matching $N = \sigma M$. Thus, when $\sigma \approx \Gamma_A$ and $M \approx \Omega|A$, we have $\Pr( \sigma M = N) = (s-1)! \times \frac{1}{|R_A|} \times \frac{1}{|A|}$. This does not depend upon $N$, and so $\sigma M$ is uniformly distributed.
  \end{proof}

\begin{proposition}
Property (C2) holds.
\end{proposition}
\begin{proof}
  Consider $A = \langle e \rangle$ where $e = \{x_1, \dots, x_s \}$ and $A' = \langle e' \rangle$ and $M \in A - A'$. We cannot have $A = A'$ since $A - A'$ is non-empty, and so $e, e'$ are disjoint.

  Suppose for contradiction that $e' \in \sigma  M$ for $\sigma \in R_A$. Let $i \geq 2$ be maximal such that $e' \in (x_i \ z_i) \cdots (x_s \ z_s) M$. We must have $i \leq s$, since $e' \notin M$. It must be the case that $z_i \in e'$. Then matching $N = (x_{i+1} \ z_{i+1}) \cdots (x_s \ z_s) M$ must contain an edge $(e' - z_i) \cup \{x_i \}$. Thus, $x_i$ is matched to the vertices $e' - z_i$ in $N$.   On the other hand, the entries $z_{i+1}, \dots, z_s$ are all distinct from $x_1, \dots, x_{i+1}$; therefore, in the matching $N$, the entries $x_1, \dots, x_{i}$ are not affected, and so $x_1, \dots, x_i$ are matched to each other.  Thus $x_i$ is matched in $N$ to $s-1$ vertices in $e'$ as well as $i-1$ vertices in $e$. Since $N$ contains only $s$-edges, this is impossible.
\end{proof}

\begin{proposition}
  \label{cyy1}
  Let $A = \langle e \rangle$ where $e = \{x_1, \dots, x_s \}$ and $A' = \langle e' \rangle$ and $M \in A \cap A'$ for $A \not \sim A'$. Consider $\sigma \in R_A$ of the form $\sigma = (x_2 \ z_2) \cdots (x_s \ z_s)$ where $z_i \in [n] - \{x_1, \dots, x_{i-1} \}$. Let $Z = \{z_2, \dots, z_s \}$.
    \begin{enumerate}
  \item If $A = A'$, then $\sigma M \in A' \Leftrightarrow Z \subseteq e'$ 
  \item If $A \neq A'$, then $\sigma M \in A' \Leftrightarrow Z \cap e' = \emptyset$.
  \end{enumerate}
\end{proposition}
\begin{proof}
For case (1),  suppose $Z \subseteq e' = e$. So each $(x_i \ z_i)$ permutes two elements within $e$, and thus a simple induction on $i$ shows that $(x_i \ z_i) \cdots (x_s \ z_s) M = M$ for all $i = k+1, \dots, 2$. In particular $\sigma M = M$. On the other hand, let $i$ be maximal such that $z_i \notin e$. Then $(x_{i+1} \ z_{i+1}) \cdots (x_s \ z_s) M = M$. This $z_i$ will remain matched to $x_1$ in $(x_2 \ z_2) \cdots (x_s \ z_s) M$, and in particular $e \notin (x_2 \ z_2) \cdots (x_s \ z_s) M$.

For case (2), we have $e \cap e' = \emptyset$ since $A \not \sim A'$. If $Z \cap e' = \emptyset$, then edge $e'$ is unaffected in $(x_2 \ z_2) \cdots (x_s \ z_s) M$, and so $e' \in M$.  On the other hand, let $i$ be maximal such that $z_i \in e'$. This $z_i$ remains matched to $x_1$ in $(x_2 \ z_2) \cdots (x_s \ z_s) M$, and in particular the edge $e'$ cannot remain in $(x_2 \ z_2) \cdots (x_s \ z_s) M$.
  \end{proof}

\begin{proposition}
  Property (C4) holds.
\end{proposition}
\begin{proof}
  Proposition~\ref{cyy1} gives an explicit condition on when $\sigma M \in A'$ for $A \not \sim A', M \in A \cap A', \sigma \in R_A$. This condition depends solely on $A, A', \sigma$ and not on $M$.
  \end{proof}

\begin{proposition}
  Property (D1') holds.
\end{proposition}
\begin{proof}
Consider $E = \{ \langle e_1 \rangle, \dots, \langle e_k \rangle \}$ and $A = \langle e \rangle$ where $e = \{x_1, \dots, x_s \}$. If $e_1, \dots, e_k$ are distinct from $e$, then we can sample $\sigma = (x_2 \ z_2) \dots (x_s \ z_s) \approx \Gamma_{A; E}$ by selecting each $z_i$ independently from the set $[n] - (e_1 \cup \dots \cup e_k) - \{x_1, \dots, x_{i-1} \}$. Similarly, if one of the sets $e_i$ is equal to $e$, then we select $z_i$ independently from $e - \{x_1, \dots, x_{i-1} \}$.
\end{proof}

\begin{proposition}
For $s = 2$, property (C3') holds.
\end{proposition}
\begin{proof}
  Consider $A_1 = \langle(x_1, y_1) \rangle, A_2 = \langle (x_2, y_2) \rangle$ and a matching $M \supseteq \{ \{x_1, y_1 \}, \{x_2, y_2 \} \}$. We need to show that for any $z_1 \in [n] - \{x_1 \}, z_2 \in [n] - \{x_1, y_1, x_2 \}$ there are $z'_2 \in [n] - \{ x_2 \}$ and $z'_1 \in [n] - \{x_2, y_2, x_1 \}$ such that
  \begin{equation}
    \label{teqn1}
    (y_1 \ z_1) (y_2 \ z_2) M = (y_2 \ z'_2) (y_1 \ z'_1) M
    \end{equation}

  By relabeling, we assume without loss of generality that  $x_1 = 1, y_1 = 3, x_2 = 2, y_2 = 4$, and $z_1, z_2 \in \{1, \dots, 6 \}$, and that either  $M = \{ \{1, 3 \}, \{2, 4 \}, \{5, 6 \} \}$ or $M = \{ \{1, 3 \}, \{ 2, 4 \}, \{5, 7 \}, \{6, 8 \} \}$. We have exhaustively tested all choices $z_1, z_2$ in both cases, verifying that there is always a choice of $z'_1, z'_2 \in \{1, \dots, 8 \}$ satisfying Eq.~(\ref{teqn1}).
\end{proof}


\begin{thebibliography}{1}
\bibitem{achlioptas}
Achlioptas, D., Iliopoulos, F.: Random walks that find perfect objects and the Lov\'{a}sz Local Lemma. Journal of the ACM 63(3), Article \#22 (2016)

\bibitem{achlioptas2}
  Achlioptas, D., Iliopoulos, F.: Focused stochastic local search and the Lov\'{a}sz local lemma. Proc. 27th ACM-SIAM Symposium on Discrete Algorithms (SODA), pp. 20248-2038 (2016)
  
\bibitem{sinclair}
  Achlioptas, D., Iliopoulos, F., Sinclair, A.: Beyond the Lov\'{a}sz Local Lemma: point to set correlations and their algorithmic applications. Proc. 60th IEEE Symposium on Foundations of Computer Science (FOCS), pp. 725-744 (2019)
    
\bibitem{albert}
Albert, M., Frieze, A., Reed, B.: Multicoloured Hamilton Cycles. The Electronic Journal of Combinatorics 2(1), R10 (1995)

\bibitem{bissacot}
Bissacot, R., Fernandez, R., Procacci, A., Scoppola, B.: An improvement of the Lov\'{a}sz Local Lemma via cluster expansion. Combinatorics, Probability and Computing 20(5), pp.\ 709-719 (2011)

\bibitem{greedy2}
Blelloch, G., Fineman, J., Shun, J.: Greedy sequential maximal independent set and matching
are parallel on average. Proc. 24th ACM Symposium on Parallelism in Algorithms and Architectures (SPAA), pp. 308-317 (2012)

\bibitem{brandt}
Brandt, S., Fischer,  O., Hirvonen,  J., Keller,  B., Lempi\"{a}inen, T., Rybicki, J., Suomela, J., Uitto, J.: A lower bound for the distributed Lov\'{a}sz Local Lemma. Proc. 48th ACM Symposium on Theory of Computing (STOC), pp. 479-488 (2015)

\bibitem{pettie}
Chung, K., Pettie, S., Su, H.: Distributed algorithms for the Lov\'{a}sz local lemma and graph coloring. Distributed Computing 30(4), pp. 261-2680 (2017)

\bibitem{lfmis-p}
Cook, S.: A taxonomy of problems with fast parallel algorithms. Information and Control 64(1--3), pp.\ 2-22 (1985)

\bibitem{erdos-spencer}
Erd\H{o}s, P., Spencer, J.: Lopsided Lov\'{a}sz Local Lemma and Latin transversals. Discrete Applied Math 30(2, 3), pp. 151-154 (1990)

\bibitem{mohsen2}
Fischer, M., Ghaffari, M.: Sublogarithmic distributed algorithms for Lov\'{a}sz Local lemma, and the complexity hierarchy. Proc. 31st International Symposium on Distributed Computing (DISC), p. 18 (2017)

\bibitem{fischer}
Fischer, M., Noever, A.: Tight analysis of randomized greedy MIS. ACM Transactions on Algorithms 16(1), Article \#6 (2019)

\bibitem{hgk}
  Ghaffari, M., Harris, D., Kuhn, F.: On derandomizing local distributed algorithms. Proc. 59th IEEE Symposium on Foundations of Computer Science (FOCS), pp. 662-673 (2018)

\bibitem{graf}
Graf, A., Haxell, P.: Finding independent transversals efficiently. arXiv:1811.02687 (2018)

\bibitem{graf2}
  Graf, A., Harris, D., Haxell, P.: Algorithms for weighted independent transversals and strong colouring. arXiv:1907.00033 (2019)
  
\bibitem{hh}
  Haeupler, B., Harris, D.: Parallel algorithms and concentration bounds for the Lov\'{a}sz Local Lemma via witness DAGs. ACM Transactions on Algorithms 13(4), Article \#53 (2017)

\bibitem{hss}
Haeupler, B., Saha, B., Srinivasan, A.: New constructive aspects of the Lov\'{a}sz Local Lemma. Journal of the ACM 58(6) (2011)

\bibitem{harris-llll}
Harris, D.: Lopsidependency in the Moser-Tardos framework: beyond the Lopsided Lov\'{a}sz Local Lemma. ACM Transactions on Algorithms 13(1), Article \#17 (2016)

\bibitem{harris-mt-dist}
Harris, D.: New bounds for the Moser-Tardos distribution. Random Structures \& Algorithms 57(1), pp. 97-131 (2020)
  
  \bibitem{harris-det}
  Harris, D.: Deterministic algorithms for the Lov\'{a}sz Local Lemma: simpler, more general, and more parallel. arXiv:1909.08065 (2019)
  
\bibitem{harris-srin2}
Harris, D., Srinivasan, A.: Algorithmic and enumerative aspects of the Moser-Tardos distribution. ACM Transactions on Algorithms 13(3), Article \#33 (2017)

\bibitem{harris-srin-perm}
  Harris, D., Srinivasan, A.: A constructive Lov\'{a}sz Local Lemma for permutations. Theory of Computing 13(17), pp. 1-41 (2017)

\bibitem{harris-srin-partial}
  Harris, D., Srinivasan, A.: The Moser-Tardos framework with partial resampling. Journal of the ACM 66(5), Article \#36 (2019)

\bibitem{harvey2}
Harvey, N., Liaw, C.: Rainbow Hamilton cycles and lopsidependency. Discrete Mathematics 340(6), pp. 1261-1270 (2017)

\bibitem{harvey}
Harvey, N., Vondr\'{a}k, J.: An algorithmic proof of the Lopsided Lov\'{a}sz Local Lemma via resampling oracles. SIAM Journal on Computing 49(2), pp. 394-428 (2020)

\bibitem{haxell2008}
Haxell, P.: An improved bound for the strong chromatic number. Journal of Graph Theory 58(2), pp.\ 148-158 (2008)

\bibitem{ilio}
Iliopoulos, F.: Commutative algorithms approximate the LLL distribution. Approximation, Randomization, and Combinatorial  Optimization. Algorithms and Techniques (APPROX/RANDOM), pp. 44:1--44:20 (2018)

\bibitem{kolipaka}
Kolipaka, K., Szegedy, M.: Moser and Tardos meet Lov\'asz. Proc. 43rd ACM Symposium on Theory of Computing (STOC), pp. 235-244 (2011)

\bibitem{kolmogorov}
Kolmogorov, V.: Commutativity in the algorithmic Lov\'{a}sz Local Lemma. SIAM Journal on Computing 47(6), pp. 2029-2056 (2018) 

\bibitem{mohr}
  Lu, L., Mohr, A., Sz\'{e}kely, L.: Quest for negative dependency graphs. Recent Advances in Harmonic Analysis and Applications, pp. 243-256 (2012)
  
\bibitem{random-inj}
Lu, L., Sz\'{e}k\'{e}ly, L.: Using Lov\'{a}sz Local Lemma in the space of random injections. The Electronic Journal of Combinatorics 13-R63 (2007)

\bibitem{lu-szekeley2}
Lu, L., Sz\'{e}k\'{e}ly, L.: A new asymptotic enumeration technique: the Lov\'{a}sz local lemma. arXiv:0905.3983 (2011)

\bibitem{luby-mis}
Luby, M.: A simple parallel algorithm for the maximal independent set problem. SIAM Journal on Computing 15(4), pp. 1036-1053 (1996)

\bibitem{mcdiarmid}
McDiarmid, C.: Hypergraph coloring and the Lov\'{a}sz Local Lemma. Journal of Discrete Mathematics 167-168, pp. 481-486 (1995)

\bibitem{moser-tardos}
Moser, R., Tardos, G.: A constructive proof of the general Lov\'{a}sz Local Lemma. Journal of the ACM 57(2), Article \#11 (2010)

\bibitem{pegden}
Pegden, W.: An extension of the Moser-Tardos algorithmic Local Lemma. SIAM Journal on Discrete Mathematics 28(2), pp. 911-917 (2014)

\bibitem{shearer}
Shearer, J. B.: On a problem of Spencer. Combinatorica 5(3), pp. 241-245 (1985)
\end{thebibliography}
\end{document}